\documentclass[lettersize,journal]{IEEEtran}
\usepackage{amsmath,amsfonts}
\usepackage{amssymb}
\usepackage{algorithmic}
\usepackage[ruled,vlined]{algorithm2e}
\usepackage{array}
\usepackage{xcolor}
\usepackage[caption=false,font=normalsize,labelfont=sf,textfont=sf]{subfig}
\usepackage{stfloats}
\usepackage{url}
\usepackage{verbatim}
\usepackage{graphicx}
\usepackage{color}
\usepackage{caption}
\usepackage{epstopdf}
\usepackage{cite}
\usepackage{bm}
\begin{document}
\newtheorem{theorem}{Theorem}
\newtheorem{remark}{Remark}
\newenvironment{proof}{{\indent \it Proof:}}{\hfill $\blacksquare$\par}
\title{Cram{\'e}r-Rao Bound Analysis and Beamforming Design for Integrated Sensing and Communication with Extended Targets}

%Integrated Sensing and Communication Beamforming Design: CRB Optimization for Extended Target
%\author{IEEE Publication Technology,~\IEEEmembership{Staff,~IEEE,}
        % <-this % stops a space
%\thanks{This paper was produced by the IEEE Publication Technology Group. They are in Piscataway, NJ.}% <-this % stops a space
%\thanks{Manuscript received April 19, 2021; revised August 16, 2021.}}


\author{Yiqiu~Wang,~\IEEEmembership{Graduate~Student~Member,~IEEE, }Meixia~Tao,~\IEEEmembership{Fellow,~IEEE, }and~Shu~Sun,~\IEEEmembership{Member,~IEEE}\vspace{-1em} % 移到这里可能有助于避免空白页
\thanks{The authors are with the Department of Electronic Engineering and the Cooperative Medianet Innovation Center (CMIC), Shanghai Jiao Tong University, China (e-mail: wyq18962080590, mxtao, shusun@sjtu.edu.cn).

Part of this work was presented in the 8th workshop on integrated sensing and communications for internet of things at the IEEE Global Telecommunications Conference (GLOBECOM), Kuala Lumpur, Malaysia, Dec. 2023 \cite{ref1}.} }
\maketitle

\makeatletter
\newcommand{\linebreakand}{%
  \end{@IEEEauthorhalign}
  \hfill\mbox{}\par
  \mbox{}\hfill\begin{@IEEEauthorhalign}
}
\makeatother

\begin{abstract}
This paper studies an integrated sensing and communication (ISAC) system, where a multi-antenna base station transmits beamformed signals for joint downlink multi-user communication and radar sensing of an extended target (ET). By considering echo signals as reflections from valid elements on the ET contour, a set of novel Cramér-Rao bounds (CRBs) is derived for parameter estimation of the ET, including central range, direction, and orientation. The ISAC transmit beamforming design is then formulated as an optimization problem, aiming to minimize the CRB associated with radar sensing, while satisfying a minimum signal-to-interference-pulse-noise ratio requirement for each communication user, along with a 3-dB beam coverage constraint tailored for the ET. To solve this non-convex problem, we utilize semidefinite relaxation (SDR) and propose a rank-one solution extraction scheme for non-tight relaxation circumstances. To reduce the computation complexity, we further employ an efficient zero-forcing (ZF) based beamforming design, where the sensing task is performed in the null space of communication channels. Numerical results validate the effectiveness of the obtained CRB, revealing the diverse features of CRB for differently shaped ETs. The proposed SDR beamforming design outperforms benchmark designs with lower estimation error and CRB, while the ZF beamforming design greatly improves computation efficiency with minor sensing performance loss.
\end{abstract}

\begin{IEEEkeywords}
 Integrated sensing and communication, Cramér-Rao bound, transmit beamforming design, semidefinite relaxation, zero-forcing.
\end{IEEEkeywords}

\section{Introduction}
\IEEEPARstart{C}{ommunication} and sensing share similar development trends, including higher frequency bands, larger antenna arrays, and hardware miniaturization. Integrated sensing and communication (ISAC) \cite{ref2,ref3,ref4,ref5}, which aims to enable both functionalities in the same frequency band, has become one of the key usage scenarios of 6G. Among various ISAC systems, the combination of radar sensing and communication \cite{ref6,ref7,ref8,ref9,ref10,ref11} is a common form. However, the fundamental goals to be achieved by radar and communication systems are completely different. Traditional radar designs focus on extracting target information from the reflected echo signals, whereas communication aims to transmit information accurately at minimal cost. As such, a major design issue in ISAC is to accommodate the diverse design objectives of communication and radar sensing, where beamforming is considered as a critical technique. Existing ISAC waveform design can be categorized into three major classes: 1) radar-centric design, 2) communication-centric design, and 3) joint design.

Radar-centric designs embed communication signals into radar probing signals. Typical methods include pulse interval modulation (PIM), index modulation, and beampattern modulation \cite{ref12,ref13,ref14}. PIM \cite{ref12} dates back to the early works of missile range instrumentation radars, where certain pulses are selected to establish a one-way communication system. Index modulation \cite{ref13} generalizes PIM by utilizing multiple waveform features, for instance subcarriers, spreading codes, and polarization states, as indices to modulate communication data. For beampattern modulation \cite{ref14}, the sensing functionality is guaranteed by the main beam which directly points to the radar target, whereas communication signals are embedded in the variation of radar beampattern sidelobes. As such, for radar-centric designs, we can observe a good preservation of the primary sensing functionality, whereas the communication performance, for instance data rate and error rate, can barely support practical demands.

On the other hand, communication-centric designs utilize existing communication waveforms to perform radar sensing tasks, among which orthogonal frequency division multiplexing (OFDM) \cite{ref15} and orthogonal time frequency space (OTFS) \cite{ref16} are deemed as two promising candidates. For instance, a two-phase scheme is proposed in \cite{ref15} for device-free sensing in OFDM cellular networks, where the delay and range of the target are estimated based on the reflected OFDM signals through base station \textcolor{black}{(BS)}-target-BS paths. The authors in \cite{ref16} employ OTFS signals for target detection, where the target range and velocity can be effectively obtained via maximum likelihood estimator in the delay-Doppler domain. Nevertheless, for communication-centric designs, the randomness of communication signals inevitably leads to degraded sensing performance.

To address the aforementioned issues, recently there have been many research efforts \cite{ref17,ref18,Hua23TVT,Liu20TSP,ref19} devoted to the beamforming design for joint radar sensing and communication. One typical approach is to take into account both communication and radar sensing metrics in the beamforming optimization problem, thereby achieving performance tradeoff for communication and sensing through the scheduling of wireless resources. While data rate and bit error rate are commonly employed to evaluate the communication performance, radar performance metrics are quite diversified, including mutual information \cite{ref17}, radar beampattern similarity\textcolor{black}{\cite{ref18,Hua23TVT,Liu20TSP}}, sidelobe-to-mainlobe ratio \cite{ref19}, and Cramér-Rao bound (CRB) \cite{ref20,Liu21TSP,ref22,Garcia22TSP}. In this work, we use the widely acknowledged CRB as the sensing performance indicator, which defines the lower bound on the variance of any unbiased estimator and serves as a benchmark for the evaluation and design of ISAC systems.

The authors in \cite{ref20,Liu21TSP} derive the CRB on the direction of arrival (DoA) estimation of the point radar target by extracting sufficient statistic variables from the echo signals. However, it is usually implicit to represent the sensing target as a single point, i.e. point target (PT), in a similar manner to the modeling of a communication user (CU), especially when the target is located near the BS with a significant physical size. In this spirit, extended target (ET) is increasingly accepted by the ISAC community with specifically-derived CRB formulas and beamforming designs. For instance, suppose there exists line-of-sight (LoS) between the ET and the BS, the authors in \cite{Liu21TSP} estimate the whole response matrix and obtain a closed-form CRB expression as a function of the beamforming matrix. The authors in \cite{ref22} consider the \textcolor{black}{non-line-of-sight} (NLoS) ISAC scenario and derive the response matrix estimation CRB for reconfigurable intelligent surface aided sensing. Although further information, such as target range or direction, can be extracted from the response matrix with sophisticated signal processing algorithms, it is unintuitive and imprecise to use CRB on the intermediate response matrix to represent estimation accuracy of the true desired parameters.

In regard of the problem above, the authors in \cite{Garcia22TSP} analyze the CRBs for the explicit estimation of the central range, direction, and orientation of an ET with both known and unknown contours. Nevertheless, this work merely considers the single-input multiple-output (SIMO) scenario with the radar sensing task only. Furthermore, the explicit CRB formula is also too complicated to construct a solvable optimization problem, especially when the communication function is integrated with the sensing task. To the best of our knowledge, no prior work has studied the CRB of direct estimation for ET parameters (i.e. range, direction, and orientation) or the associated optimization in the multi-user multiple-output multiple-input (MU-MIMO) ISAC scenario.

In this paper, we propose a transmit beamforming design framework for MU-MIMO ISAC systems, by specially taking into account the optimization of CRB for ET sensing. The ISAC system consists of one MIMO BS with a uniform linear array (ULA), one ET modelled by a truncated Fourier series (TFS) contour, and multiple CUs. We assume that the channel state information, the transmitted signal, and DoAs of CUs are perfectly known by the BS for downlink communication, while the direction of the ET is to be estimated at the BS using echo signals. We aim to minimize the CRB of ET direction estimation by optimizing the transmit beamforming at the BS, under certain communication and beampattern constraints.

The main contributions of this work are summarized as follows:

\begin{itemize}
    \item First, we characterize, in closed-form expressions, the CRBs for the range, direction, and orientation estimation performance of the ET. Compared with \cite{Garcia22TSP}, our CRB analysis considers a more general MIMO ISAC scenario with the transmit beamforming. Theoretical analysis unveils the explicit dependence of CRBs upon sensing path loss, noise power, the number of transceiver antennas, and signal covariance matrix. We also establish an interesting connection between the CRB of ET and that of PT.
    \item Second, we formulate a CRB minimization problem by optimizing the transmit beamforming, subject to the minimum signal-to-interference-plus noise ratio (SINR) constraint for all CUs along with a distinct beam coverage constraint for the ET. Due to the non-convex and quadratic structure of the CRB, we propose a semidefinite relaxation (SDR) based algorithm to obtain the near-optimal solution, which also consists of a novel procedure to extract rank-one beamformers for non-tight relaxation scenarios.
    \item Third, to relieve the computational burden of SDR beamforming, we propose a low-complexity beamforming algorithm based on zero-forcing (ZF) the interference between CUs. By splitting the beamforming matrix into the pseudoinverse and null space components of the communication channel, we align the null space components towards varied points on the ET contour to perform sensing task. Consequently, the matrix variables in the SDR problem degrade into vector variables, leading to an efficient ZF beamforming algorithm.
    \item Finally, we numerically analyze the CRB of ET and evaluate the performance of the proposed beamforming designs. The diverse characteristics of CRB are explored for ETs with different shapes. Further, compared with existing beamforming designs, our proposed CRB-based design effectively boosts sensing performance in both estimation mean square error (MSE) and CRB. Moreover, the ZF design trades minor CRB loss for great efficiency improvement.
\end{itemize}

The remainder of this paper is organized as follows. Section II introduces the ISAC system model and the TFS contour model for ET. Section III presents the derivation of CRBs. Section IV proposes the ISAC beamforming optimization problem, together with the distinct SDR and ZF beamforming design algorithms. Section V outlines the numerical results. Finally, Section VI concludes the paper.

\textit{Notations:} $\left[\cdot\right]^T$, $\left[\cdot\right]^H$, $\left[\cdot\right]^*$ denote, respectively, the transpose, Hermitian transpose, and conjugate of a matrix; $\mathbb{E}\left[\cdot\right]$ denotes the mean of variables; $\mathbf{0}_{m \times n}$ and $\mathbf{I}_{m}$ denote an ${m \times n}$ matrix with all zero elements, and an ${m \times m}$ identity matrix, respectively; $\mathcal{R\left(\cdot\right)}$ and $\mathcal{I\left(\cdot\right)}$ respectively denote the real and imaginary part of a complex number; $\mathcal{CN}\left(\mathbf{0}_{m\times 1},\sigma^{2}\mathbf{I}_m\right)$ denotes the probability density of an ${m\times 1}$ circularly symmetric complex Gaussian vector with zero mean and covariance matrix $\sigma^{2}\mathbf{I}_m$; $\mathcal{U}(-\pi/2,\pi/2)$ denotes the uniform distribution over the interval from $-\pi/2$ to $\pi/2$; $\mathbb{R}^{m\times n}$ and $\mathbb{C}^{m\times n}$ denote a matrix with ${m\times n}$ real elements and ${m\times n}$ complex elements, respectively; $\Delta _{{{\boldsymbol{\theta} }_{1}}}^{{\boldsymbol{\theta }_{2}}}f\left[ \cdot \right]=\frac{\partial}{\partial\boldsymbol{\theta}_1}\frac{\partial}{\partial\boldsymbol{\theta}_2^T}f\left[ \cdot \right]$ denotes the second derivative over ${\boldsymbol{\theta} }_{1}$ and ${\boldsymbol{\theta}}_{2}$; $\mathbf{A} \succeq 0$ denotes that matrix $\mathbf{A}$ is semi-definite; $\Vert\cdot\Vert$ denotes the $l_2$ norm.

\section{System Model}
As depicted in Fig. \ref{Fig1}, we consider a downlink ISAC system. A monostatic MIMO radar and an MIMO communication transmitter are integrated inside a BS, which is equipped with {$N_t$} transmit antennas and {$N_r$} receive antennas. The BS sends wireless signals to perform radar sensing of an ET and data communication with $N_c$ single-antenna CUs simultaneously. The CUs are capable of decoding communication messages based on their own received signals. At the same time, the BS receives echo signals reflected from the surface of the ET, from which the unknown parameters of the ET, including the range, direction, and orientation, can be extracted with specific
methods. To guarantee the feasibility of basic radar sensing and communication functions, we assume $N_c \leq N_t \leq N_r$ throughout the paper.

In the sequel, we first introduce the signal model for communication, followed by the description of the contour modeling for ET, as well as the sensing signal model.

\begin{figure}[!t]
\centering
\includegraphics[width=3in]{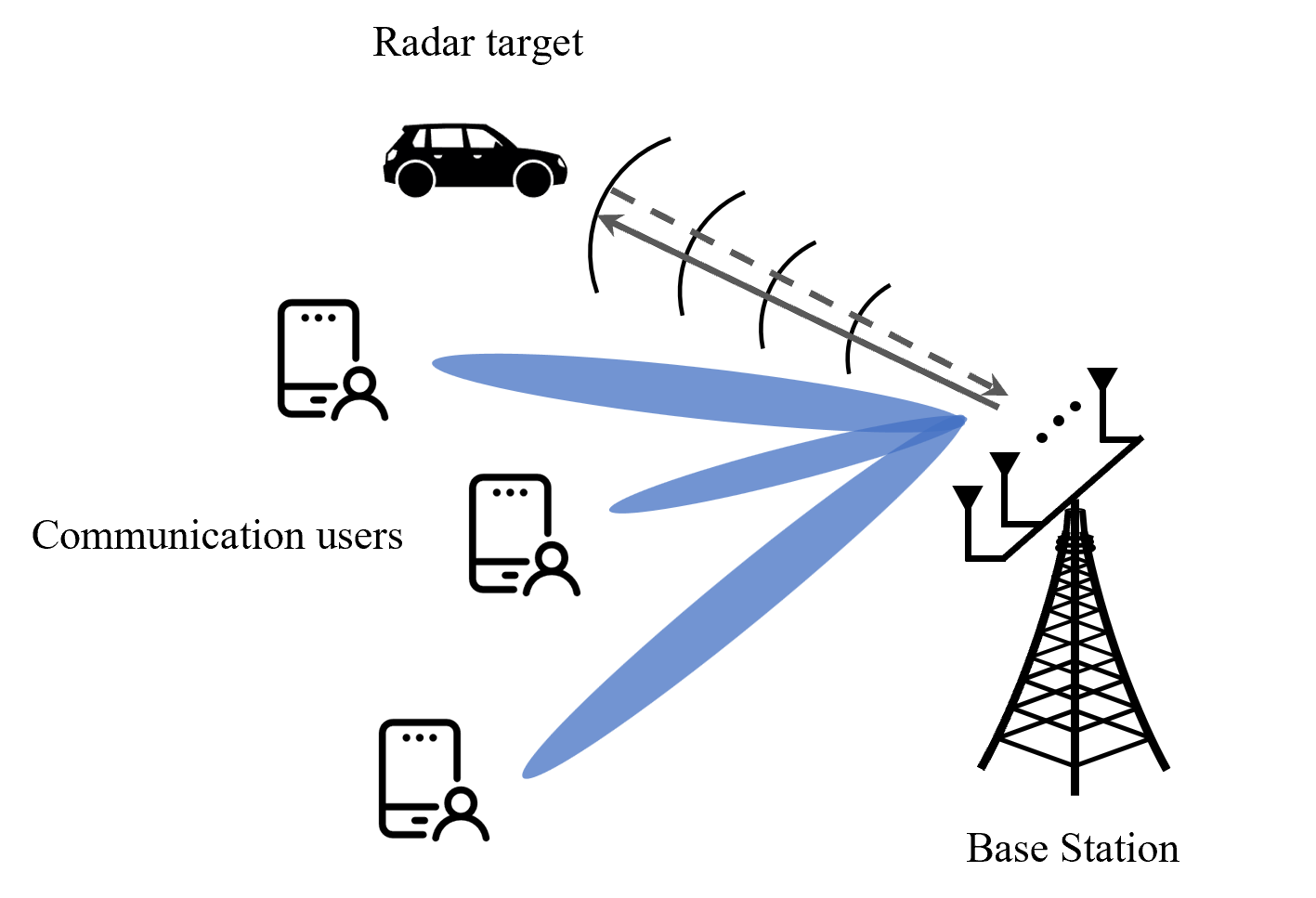}
\caption{\textcolor{black}{Monostatic ISAC system in which one BS performs communication and sensing tasks simultaneously.}}
\label{Fig1}
\vspace{-10pt}
\end{figure}

\subsection{Communication Signal Model}
The BS performs sensing based on existing communication protocols, where all transmit antennas are utilized to send only communication signals for ISAC tasks. The transmit signal can be expressed as
\begin{equation}
\mathbf{x}(t) = \sum\nolimits_{n=1}^{N_c}\mathbf{w}_{n}{c}_{n}(t),
\end{equation}
where $\mathbf{w}_n \in \mathbb{C}^{N_{t}}$ and ${c}_{n}(t)$ are the beamforming vector and communication signal for the $n$-th CU, respectively. The communication signals are generated by wide-sense stationary stochastic process with zero-mean and unit power, and they are uncorrelated for different CUs. That is,
\begin{equation}
\label{eq2}
\mathbb{E} \left[ \mathbf{c}(t) \mathbf{c}^H (t) \right] = \mathbf{I}_{N_c},
\end{equation}
where $\mathbf{c}(t) = \left[{c}_{1}(t),...,{c}_{N_c}(t) \right]^{T} \in \mathbb{C}^{N_c}$ is the communication symbol vector. Let $\mathbf{W} = \left[\mathbf{w}_1,...,\mathbf{w}_{N_c} \right] \in \mathbb{C}^{N_{t}\times N_c}$ denote the transmit beamforming matrix. The covariance matrix of transmitted signals is
\begin{equation}
\label{eq_3}
\mathbf{R}_x = \mathbb{E} \left[\mathbf{x}(t) \mathbf{x}^H (t)\right] = \mathbf{W} \mathbf{W}^H.
\end{equation}

The received signal $y_n(t)$ at the $n$-th CU is expressed as
\begin{equation}%会产生编号
y_n(t) = \mathbf{h}_n^H \mathbf{x}(t) + {z}_{n}(t) = \mathbf{h}_n^H \sum\nolimits_{n=1}^{N_c} \textcolor{black}{\mathbf{w}_{n}}c_{n}(t) + {z}_{n}(t),
\end{equation}
where ${z}_n(t)$ is the additive white Gaussian noise (AWGN) with zero-mean and $\sigma_c^2$ variance, and $\mathbf{h}_n \in \mathbb{C}^{N_t}$ is the communication channel vector. We adopt the narrow-band assumption since the frequency selectivity of the channel is negligible when $B \ll f_c$, where $B$ is the effective bandwidth of the ISAC system, $f_c$ is the carrier frequency. By following the Saleh-Valenzuela model, the channel vector can be expressed as

\begin{equation}
\mathbf{h}_n = \sqrt{g_n} \sum\nolimits_{l=1}^L \beta_{l,n}\mathbf{a}\left(\phi_{l,n}^t\right),
\end{equation}
where $g_n$ is the large-scale path loss, $L$ denotes the number of multipath components between the CU and the BS, $\beta_{l,n}\sim\mathcal{CN}(0,\sigma_{l,n}^2)$ is the complex gain of the $l$-th path with the assumption of $\sigma_{1,c}^2 \geq \sigma_{2,n}^2 \geq ... \geq \sigma_{L,n}^2$ and $\sum_{l=1}^L \sigma_{l,n}^2=1$, $\phi_{l,n}^t \sim \mathcal{U}(-\pi/2,\pi/2)$ represent the directions of departure at the BS.

To guarantee the communication quality of each CU, the transmit beamformers should be designed to reduce the inter-user interference among different CUs, so as to achieve a certain level of SINR for all users. The SINR of the $n$-th CU can is given by

\begin{equation}
{{\gamma }_{n}}=\frac{{\left| \mathbf{h}_{n}^{H}{{\mathbf{w}}_{n}} \right|}^{2}}{\sum\nolimits_{i=1,i\ne n}^{N_c}{{{\left| \mathbf{h}_{n}^{H}{{\mathbf{w}}_{i}} \right|}^{2}}+\sigma _{n}^{2}}}.
\end{equation}

The achievable sum-rate of all CUs can be obtained as $\sum\nolimits_{n=1}^{N_c}{{{\log }_{2}}\left( 1+{{\gamma }_{n}} \right)}$.

\subsection{Extended Target Contour Model}
\begin{figure}[!t]
\centering
\includegraphics[width=3.5in]{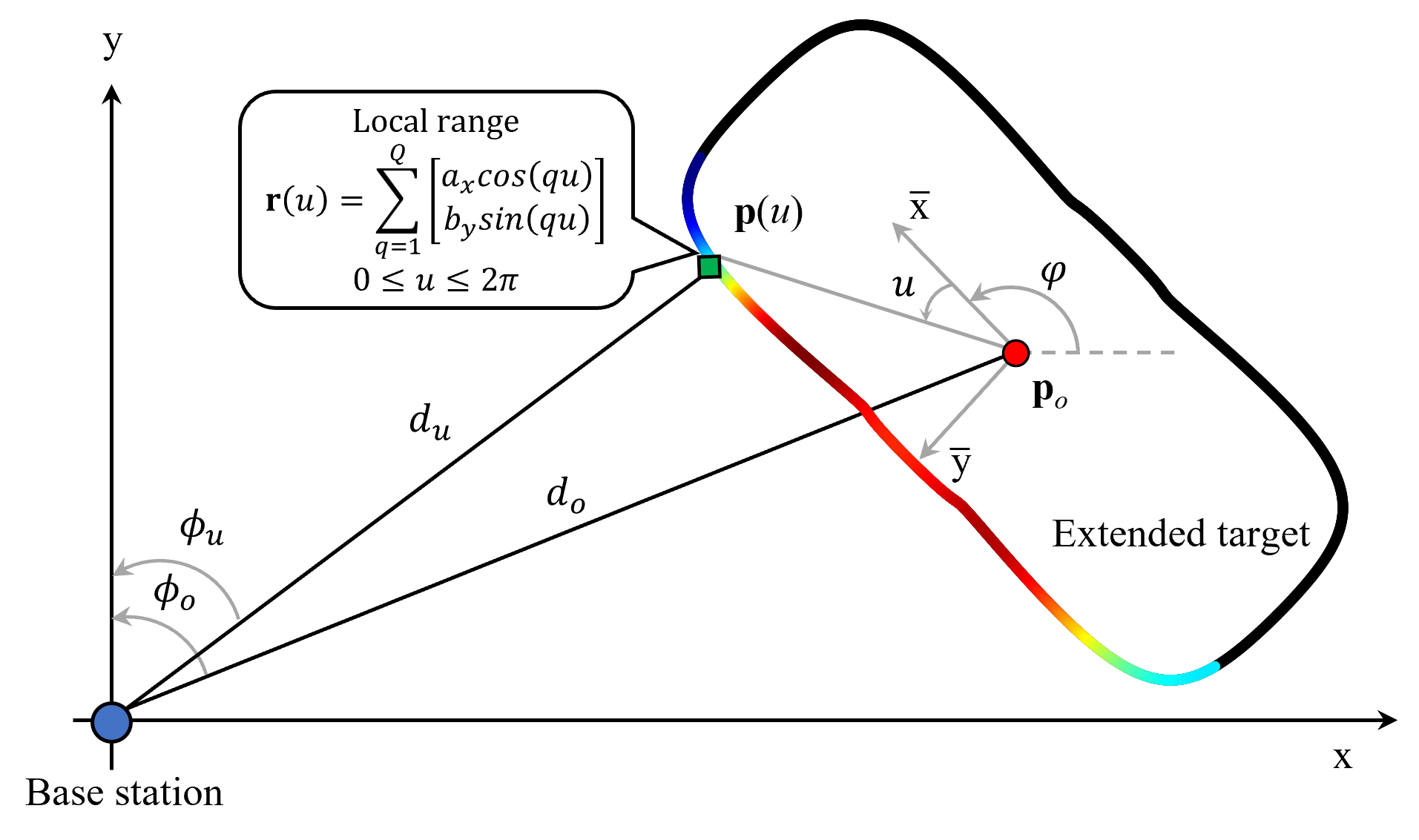}
\caption{Bird-eye view of an ET with TFS contour \cite{Garcia22TSP}.}
\label{Fig2}
\end{figure}

In this work, the received echo signals at the BS are assumed to be reflected from the ET contour elements with LoS to the BS, and the reflection from interior components of the ET is ignored considering the significant penetration loss. Consequently, the choice of contour modeling directly affects the sensing performance. Various ET contour models have been proposed in the radar sensing literature, e.g. multiple ellipses \cite{ref24,ref25}, B-spline curve \cite{ref26}, and Gaussian process \cite{ref27}. These models can capture useful details such as the rounded corners of ET, yet typically come at the cost of notably increased complexity and interpretability lackage with respect to the ET parameters. Instead, in this work we build on the TFS model proposed in \cite{Garcia22TSP,ref28}, which utilizes a tractable number of parameters to describe the ET contour and effectively captures the backscattering effects generated by a waveform impinging on an arbitrarily shaped ET.

As presented in the bird-eye-view of Fig. \ref{Fig2}, the ET is located at the origin of the local coordinate system where the $+\overline{x}$ axis is aligned with the heading of the target. The target orientation is defined as the angle $\varphi$ from the global $+x$ axis to the local $+\overline{x}$ axis. The local coordinate of each contour element is represented by the TFS with $2Q$ coefficients as a function of local direction $u$. Particularly, while the general TFS has both sine and cosine harmonics, we project each contour element onto the local $\overline{x}$ and $\overline{y}$ axis, and represent the corresponding $x$ and $y$ components in the local coordinate with only cosine and sine harmonics, respectively. The above modelling scheme can effectively reduce the number of necessary contour coefficients, especially for targets with a symmetric contour. As such, the contour can be generated for local direction $u\in \left[0,2\pi\right]$, by
\begin{equation}
\boldsymbol{\rho}\left( u \right)=\sum\limits_{q=1}^{Q}{\left[ \begin{matrix}
    {{a}_{q}}\cos \left( qu \right) \\ 
    {{b}_{q}}\sin \left( qu \right) \\ 
\end{matrix} \right]}=\left[ \begin{matrix}
   {\textcolor{black}{{\boldsymbol{\nu }}}^{T}}\mathbf{m}  \\
   {{\boldsymbol{\varsigma }}^{T}}\mathbf{n}  \\
\end{matrix} \right],
\end{equation}
where $\boldsymbol{\nu }={{\left[ \cos \left( u \right),...,\cos \left( Qu \right) \right]}^{T}}$ is the cosine harmonics, $\boldsymbol{\varsigma }={{\left[ \sin \left( u \right),...,\sin \left( Qu \right) \right]}^{T}}$ is the sine harmonics, $\mathbf{m}={{\left[ {{a}_{1}},...,{{a}_{Q}} \right]}^{T}}$is the TFS cosine coefficient and $\mathbf{n}={{\left[ {{b}_{1}},...,{{b}_{Q}} \right]}^{T}}$ is the TFS sine coefficient. Since the target is assumed to be located at the \textcolor{black}{origin} of the local frame, the zeroth order TFS coefficients ${{a}_{0}}$ and ${{b}_{0}}$ are not considered here. To obtain an anti-clockwise cycling contour, the coefficients of the first harmonic should satisfy ${{a}_{1}}>0$, ${{b}_{1}}>0$. While the modified TFS naturally captures the symmetry characteristic of the target contour, the number of TFS coefficients $Q$ determines the precision of the curve when approximating the ground-truth contour. Typically, when the contour details are out of consideration, e.g. in bandwidth shortage and distant location circumstances, it is more appropriate to assign $Q$ with a relatively small value. On the other hand, if the target is located close to the BS, or the contour is highly irregular with twists and turns, we tend to set a large $Q$ for better approximation.

Assume the BS is located at the origin of the global coordinate system and $+x$ axis is aligned with the antenna orientation. The target displacement can be presented as
\begin{equation}
{{\mathbf{p}}_{o}}={{\left[ {{d}_{o}}\cos {{\phi }_{o}},\text{   }{{d}_{o}}\sin {{\phi }_{o}} \right]}^{T}},
\end{equation}
where ${{d}_{o}}$ represents the range between the BS and the target center and ${{\phi }_{o}}$ is the target direction. For a specific element on the contour  with local direction $u$, its global displacement can be expressed as
\begin{equation}
\label{p(u)}
\mathbf{p}\left( u \right)={{\mathbf{p}}_{o}}+\mathbf{V}\boldsymbol{\rho}\left( u \right),
\end{equation}
where $\mathbf{V=}\left[ \begin{matrix} \cos \varphi  & -\sin \varphi \\ \sin \varphi  & \cos \varphi \\ \end{matrix} \right]$ is the spin matrix at orientation $\varphi$ with regard to the $+x$ axis in the global coordinate system. Finally, we obtain the complete target contour in the global coordinate, defined as $\mathcal{C}=\left\{ \mathbf{p}\left( u \right):0\le u\le 2\pi \right\}$.

\begin{remark}
\textcolor{black}{
   In the existing ISAC literature, an ET is sometimes treated as multiple virtual anchors and therefore the sensing of an ET is simply to estimate the multiple scatterers along the ET \cite{Du23TWC}, or the corresponding sensing response matrix \cite{Liu21TSP}. Such a method cannot, however, explicitly estimate the center parameters of the ET, e.g., $d_o$, $\phi_o$ and $\varphi$, which can be further used to recover all elements along the ET contour with $(\ref{p(u)})$. When tracking multiple scatterers along the ET, the detectable scatterers of faraway ET may abnormally vanish due to limited angular resolution. However, this barely affects the estimation of the ET center, since it relies on the echo signals regarding all visible ET contour elements. This is the main advantage of using the contour model to describe an ET, rather than using multiple scatterers, for target estimation or target tracking.}
\end{remark}

\subsection{Received Sensing Signal Model}
We consider a full-duplex system. The BS sends downlink communication signals and captures the echo signals reflected from the LoS contour of ET simultaneously. \textcolor{black}{The co-located transmitter and receiver antenna arrays are assumed to be perfectly decoupled, such that the sensing performance does not suffer from self-interference (SI) in the full-duplex operation.}\footnote{\textcolor{black}{In practice, if SI cannot be canceled perfectly \cite{Barneto21IWC}, we can further reduce the impact of SI by designing appropriate beamformers, for instance nulling the residual SI \cite{Liu23JSAC} or suppressing the beampattern gain in the receiver direction \cite{tang24arXiv}, which is beyond the scope of this paper.}} We define $\mathcal{C}_{LoS}$ and ${u}_{LoS}$ as the target LoS contour and local direction of LoS contour elements, respectively. While the target contour is a continuous curve on the 2D plane, we can split $\mathcal{C}_{LoS}$ into $K$ disjoint subsections with an angular interval of $\Delta u={{u}_{LoS}}/K$, then we have ${\mathcal{C}_{LoS}}=\bigcup _{k=1}^{K}{\mathcal{C}_{k}}$ where ${\mathcal{C}_{k}}=\left\{ \mathbf{p}\left( {{u}_{k}} \right),u_{LoS}^{lower}+\left( k-1 \right)\Delta u\le {{u}_{k}}\le u_{LoS}^{lower}+k\Delta u \right\}$. The local directions $u_{LoS}^{lower}\le {{u}_{1}}\le u_{LoS}^{lower}+\Delta u\le {{u}_{2}}\le ...\le {{u}_{K}}\le u_{LoS}^{upper}=u_{LoS}^{lower}+K\Delta u$ define a non-overlapping partition of the local LoS angular range $\left[ u_{LoS}^{lower},u_{LoS}^{upper} \right]$. Consequently, the received echo signal at the BS can be expressed as
\begin{align}
\textcolor{black}{\mathbf{s}\left( t \right)} &= \int_{\mathcal{C}_{LoS}} \mathbf{s}_{\boldsymbol{\rho}}(t) \, \mathrm{d} \boldsymbol{\rho} = \sum_{k=1}^{K} \mathbf{s}_{k}(t) \nonumber \\
&= g \sum_{k=1}^{K} \sqrt{l_{k}} \alpha_{k} \mathbf{b}\left( \phi_{k} \right) \mathbf{a}^{H}\left( \phi_{k} \right) \mathbf{x}\left( t - \frac{2 d_{k}}{c_0} \right)
\label{eq_ll8}
\end{align}
where $c_0$ is the speed of light, ${{\mathbf{s}}_{\boldsymbol{\rho}}}$ and ${{\mathbf{s}}_{{k}}}$ refer to the echo signals as a function of $\boldsymbol{\rho}$ and $\mathcal{C}_k$, $\alpha_{k}$, ${{\phi }_{k}}$, ${{d}_{k}}$ and $l_k$ refer to the complex radar cross section (RCS), global direction, range, and perimeter of $\mathcal{C}_k$, respectively, $g=1/d_{o}^{2}$ is the sensing path loss, $\mathbf{a}(\cdot)$ and $\mathbf{b}(\cdot)$ are the steering vectors of transmit and receive antennas, respectively. \textcolor{black}{Since the echo signals are originated from sufficiently separated scatterers on the ET, the RCS of scatterers can be conveniently modelled as independent Gaussian variables,\footnote{\textcolor{black}{One more reasonable RCS model for the ET takes the correlation between the scatterers along the ET contour into account, which formulates a complicated RCS function of incident and scattering angles \cite{Garcia22TSP}.}} i.e., $\alpha_{k}\sim \mathcal{CN}(0,1)$.} After mixing with AWGN in the wireless channel, the received sensing signal at the BS is
\begin{equation}
\label{eq14}
{{\mathbf{y}}_{s}}\left( t \right)=\mathbf{s}\left( t \right)+{{\mathbf{z}}_{s}}\left( t \right),
\end{equation}
where ${{\mathbf{z}}_{s}}\left( t \right)\in {\mathbb{C}^{{{N}_{r}}}}$ is the AWGN vector with zero-mean and $\sigma _{s}^{2}$ variance for each element.

\section{CRB analysis}

\textcolor{black}{Following the derivation framework in \cite[Section III]{Garcia22TSP}, we extend the CRB analysis from the SIMO scenario \cite{Garcia22TSP} to the generalized MIMO scenario in this section.}

\subsection{CRB for Extended Target}
We define $\boldsymbol{\xi }={{\left[ {{\boldsymbol{\alpha }}^{T}},g, {{\boldsymbol{\kappa }}^{T}} \right]}^{T}}$ as the set of unknown parameters. Here, $\boldsymbol{\alpha }=\left[ \mathcal{R}\left( {{\alpha }_{1}},{{\alpha }_{2}},...,{{\alpha }_{K}} \right),\mathcal{I}\left( {{\alpha }_{1}},{{\alpha }_{2}},...,{{\alpha }_{K}} \right) \right]^{T}\in {\mathbb{R}^{2K}}$ is a nuisance RCS vector parameter, $g$ is the sensing path loss term as defined before, and $\boldsymbol{\kappa }={{\left[ {{d}_{o}},{{\phi }_{o}},\varphi,\mathbf{m}^T, \mathbf{n}^T \right]}^{T}}\in {\mathbb{R}^{ 2Q+3 }}$ is the deterministic vector parameter of interest. As stated in (\ref{eq14}), the received echo signals are independent complex Gaussian vectors with ${{\mathbf{y}}_{s}}\left( t \right)\sim \mathcal{CN}\left( \mathbf{s}\left( t \right),\sigma _{s}^{2}{{\mathbf{I}}_{{{N}_{r}}}} \right)$. Within a certain observation period $t_s$, the log-likelihood function for estimating $\boldsymbol{\xi }$ from ${{\mathbf{y}}_{s}}\left( t \right)$ is presented as

\begin{align}
\log p\left( {{\mathbf{y}}_{s}}|\boldsymbol{\xi } \right)= &\frac{2}{{\sigma_s^2}}\mathcal{R}\int_{t_s}{\mathbf{y}_{s}^{H}\mathbf{s}\left( t \right)\mathrm{d}t}-\frac{1}{{\sigma_s^2}}\int_{t_s}{{{\left\| \mathbf{s}\left( t \right) \right\|}^{2}}\mathrm{d}t} \nonumber\\ 
& - \frac{1}{{\sigma_s^2}}\int_{t_s}{{{\left\| \mathbf{y}\left( t \right) \right\|}^{2}}\mathrm{d}t} -t_{s}N_{r}\text{log}\left(\pi\sigma_s^{2}\right).
\end{align}

According to the definition in \cite{Garcia22TSP}, we obtain the Fisher information matrix (FIM) of all parameters in $\boldsymbol{\xi}$ as
\begin{align}
\mathbf{J}\left( \boldsymbol{\xi} \right) & = -\mathbb{E} \left[ \Delta_{\boldsymbol{\xi }}^{\boldsymbol{\xi}} \log{p} \left( {\mathbf{y}}_{s},\boldsymbol{\alpha }| g,\boldsymbol{\kappa} \right) \right] \nonumber \\ 
 & =-\mathbb{E} \left[ \Delta _{\boldsymbol{\xi }}^{\boldsymbol{\xi }}\log p\left( {{\mathbf{y}}_{s}}|{\boldsymbol{\alpha} },g,\boldsymbol{\kappa } \right) \right]-\mathbb{E} \left[ \Delta _{\boldsymbol{\xi }}^{\boldsymbol{\xi }}\log p\left( \boldsymbol{\alpha |}g,\boldsymbol{\kappa } \right) \right] \nonumber \\ 
 & =-\mathbb{E} \left[ \Delta _{\boldsymbol{\xi }}^{\boldsymbol{\xi }}\log p\left( {{\mathbf{y}}_{s}}\mathbf{|\xi } \right) \right]-\mathbb{E} \left[ \Delta _{\boldsymbol{\xi }}^{\boldsymbol{\xi }}\log p\left( \boldsymbol{\alpha } \right) \right],
\label{eq_ll4}
\end{align}
where $p\left( {{\mathbf{y}}_{s}},\boldsymbol{\alpha |}g,\boldsymbol{\kappa } \right)$ is the joint \textcolor{black}{a posteriori} probability density function of the echo signals. Note that the second equality in (\ref{eq_ll4}) holds as a direct application of the Bayes theorem, and the third equality explores the fact that RCS $\boldsymbol{\alpha }$ is a random vector independent of $g$ and $\boldsymbol{\kappa }$.

Since the received signals obey Gaussian distribution, we can obtain the expected second derivative of the log-likelihood function as
\begin{equation}
-\mathbb{E} \left[ \Delta _{{{\boldsymbol{\theta} }_{1}}}^{{\boldsymbol{\theta }_{2}}}\log p\left( {{\mathbf{y}}_{s}}\mathbf{|\xi } \right) \right]=\frac{2}{\sigma _{s}^{2}} \mathcal{R}\int_{{{t}_{s}}}{\mathbb{E}\left[ \frac{\partial {{\mathbf{s}}^{H}}}{\partial {\boldsymbol{\theta }_{1}}}\frac{\partial \mathbf{s}}{\partial \boldsymbol{\theta}_{2}^{T}} \right]\mathrm{d}t},
\end{equation}
where ${{\boldsymbol{\theta} }_{1}}$ and ${{\boldsymbol{\theta} }_{2}}$ are arbitrary variables. Once again, considering the random characteristic of RCS, we can derive $\mathbb{E}\left[ \Delta _{g}^{\boldsymbol{\alpha }}\log p\left( \boldsymbol{\alpha } \right) \right]={{\mathbf{0}}_{2K\times 1}}$ and $\mathbb{E}\left[ \Delta _{\boldsymbol{\kappa }}^{\boldsymbol{\alpha }}\log p\left( \boldsymbol{\alpha } \right) \right]={\mathbf{0}}_{2K\times \left(2Q+3\right)}$. We also have $\mathbb{E}\left[ \Delta _{g}^{\boldsymbol{\alpha }}\log p\left( {{\mathbf{y}}_{s}}|\boldsymbol{\xi } \right) \right]={{\mathbf{0}}_{2K\times 1}}$ and $\mathbb{E}\left[ \Delta _{\boldsymbol{\kappa }}^{\boldsymbol{\alpha }}\log p\left( {{\mathbf{y}}_{s}}|\boldsymbol{\xi } \right) \right]={{\mathbf{0}}_{2K\times \left(2Q+3\right)}}$ since $\mathbb{E}\left[ {{\boldsymbol{\alpha }}_{k}} \right]=0$.

Combining the above properties, the FIM can be split as
\begin{equation}
\mathbf{J}\left( \boldsymbol{\xi } \right)=\left[ \begin{matrix}
   {{\mathbf{F}}_{\boldsymbol{\alpha }}} & {{\mathbf{0}}_{2K\times \left( 2Q+4 \right)}}  \\
   {{\mathbf{0}}_{\left( 2Q+4 \right)\times 2K}} & {{\mathbf{F}}_{\left( \boldsymbol{\kappa },g \right)}}  \\
\end{matrix} \right],
\end{equation}
\begin{equation}
{{\mathbf{F}}_{\left( \boldsymbol{\kappa },g \right)}}=\left[ \begin{matrix}
   {{f}_{g}} & \mathbf{f}_{\boldsymbol{\kappa },g}^{T}  \\
   {{\mathbf{f}}_{\boldsymbol{\kappa },g}} & {{\mathbf{f}}_{\boldsymbol{\kappa }}}  \\
\end{matrix} \right],
\label{eq_ll100}
\end{equation}
where ${{\mathbf{F}}_{\boldsymbol{\alpha }}}\in {\mathbb{R}^{2K\times 2K}}$, ${{\mathbf{f}}_{\boldsymbol{\kappa }}}\in {\mathbb{R}^{\left( 2Q+3 \right)\times \left( 2Q+3 \right)}}$ and ${{\mathbf{F}}_{\left( \boldsymbol{\kappa },g \right)}}\in {\mathbb{R}^{\left( 2Q+4 \right)\times \left( 2Q+4 \right)}}$ are FIMs of the RCS parameter $\boldsymbol{\alpha }$, the deterministic vector $\boldsymbol{\kappa }$ and the combination of $\left( \boldsymbol{\kappa },g \right)$, respectively, ${{i}_{{g}}}$ is the Fisher information scalar related to the path loss term $g$, ${{\mathbf{f}}_{\boldsymbol{\kappa },g}}\in {\mathbb{R}^{2Q+3}}$ refers to the Fisher vector regarding the path loss term $g$ and parameters of interest $\boldsymbol{\kappa }$.

\begin{figure*}[t]
\begin{align}
\label{CRB d_o}
& CRB\left( {{d}_{o}} \right)={{\left( \frac{2{{g}^{2}}{{N}_{r}}{{Z}_{2}}}{\sigma _{s}^{2}} \right)}^{-1}}{{\left[ \sum\limits_{k=1}^{K}{{{l}_{k}}\mathbf{a}_{k}^{H}{{\mathbf{R}}_{x}}{{\mathbf{a}}_{k}}}-{{{\left( \sum\nolimits_{k=1}^{K}{{{l}_{k}}{{X}_{k}}\mathbf{a}_{k}^{H}{{\mathbf{R}}_{x}}{{\mathbf{a}}_{k}}} \right)}^{2}}}/{\sum\nolimits_{k=1}^{K}{{{l}_{k}}X_{k}^{2}\mathbf{a}_{k}^{H}{{\mathbf{R}}_{x}}{{\mathbf{a}}_{k}}}} \right]}^{-1}}, \\
\label{CRB phi_o}
& CRB\left( {{\phi }_{o}} \right)={{\left( \frac{2{{g}^{2}}{{N}_{r}t_s}}{\sigma _{s}^{2}} \right)}^{-1}}{\textcolor{black}{\left\{ \sum\limits_{k=1}^{K}{{{l}_{k}}\left( {{Z}_{1,k}}\mathbf{a}_{k}^{H}{{\mathbf{R}}_{x}}{{\mathbf{a}}_{k}}+\mathbf{\dot{a}}_{k}^{H}{{\mathbf{R}}_{x}}{{{\mathbf{\dot{a}}}}_{k}}\right)-\left[\sum\limits_{k=1}^{K} l_k\left(\mathbf{\dot{a}}_{k}^{H}{{\mathbf{R}}_{x}}{{\mathbf{a}}_{k}}+\mathbf{\dot{a}}_{k}^{H}{{\mathbf{R}}_{x}}{{\mathbf{a}}_{k}} \right)\right]^{2}/4{\sum\limits_{k=1}^{K}l_k\mathbf{a}_{k}^{H}{{\mathbf{R}}_{x}}{{\mathbf{a}}_{k}}} } \right\}}^{-1}},\\
\label{CRB varphi}
& CRB\left( \varphi  \right)=CRB\left( {{\phi }_{o}} \right)+{{\left( \frac{2{{g}^{2}}{{N}_{r}}{{Z}_{2}}}{\sigma _{s}^{2}} \right)}^{-1}}{{\left[ \sum\limits_{k=1}^{K}{{{l}_{k}}X_{k}^{2}\mathbf{a}_{k}^{H}{{\mathbf{R}}_{x}}{{\mathbf{a}}_{k}}}-{{{\left( \sum\nolimits_{k=1}^{K}{{{l}_{k}}{{X}_{k}}\mathbf{a}_{k}^{H}{{\mathbf{R}}_{x}}{{\mathbf{a}}_{k}}} \right)}^{2}}}/{\sum\nolimits_{k=1}^{K}{{{l}_{k}}\mathbf{a}_{k}^{H}{{\mathbf{R}}_{x}}{{\mathbf{a}}_{k}}}} \right]}^{-1}}.
\end{align}\hrule
\end{figure*}

It should be noted that while there are a total of $2Q+2K+4$ unknown parameters in $\boldsymbol{\xi }$, we are only interested in the specific estimation of $2Q+3$ parameters in $\boldsymbol{\kappa }$. Thus, the effective Fisher information matrix (EFIM) should be extracted from the overall FIM for further analysis. We commence by ignoring the FIM elements related to $\boldsymbol{\alpha }$ given its irrelevance of $\boldsymbol{\kappa }$ and $g$, and only focus on the lower-right matrix block ${{\mathbf{F}}_{\left( \boldsymbol{\kappa },g \right)}}$. Then, we invert ${{\mathbf{f}}_{\boldsymbol{\kappa }}}$ as the Schur complement of ${{f}_{g}}$ over ${{\mathbf{f}}_{\boldsymbol{\kappa },g}}$ and calculate the EFIM as
\begin{equation}
\mathbf{J}\left( \boldsymbol{\kappa } \right)={{\mathbf{f}}_{\boldsymbol{\kappa }}}-{{\mathbf{f}}_{\boldsymbol{\kappa },g}}\mathbf{f}_{\boldsymbol{\kappa },g}^{T}/{f}_{g}.
\end{equation}

With prior information of the RCS parameter $\boldsymbol{\alpha }$, the simplified EFIM $\mathbf{J}\left( \boldsymbol{\kappa } \right)$ can be extracted from the original FIM $\mathbf{J}\left( \boldsymbol{\xi } \right)$. Nevertheless, we still need to calculate the matrix inverse of EFIM to obtain the final CRB matrix on $\boldsymbol{\kappa }$. It should be noted that EFIM $\mathbf{J}\left( \boldsymbol{\kappa } \right)$ has a dimension of $\left( 2Q+3 \right)\times \left( 2Q+3 \right)$, where the value of TFS parameter $Q$ is generally greater than 8 for an acceptable representation of the contour. Hence, considering the matrix inversion operation, it is hard to obtain closed-form CRB on a specific parameter directly from $\mathbf{J}\left( \boldsymbol{\kappa } \right)$, say the target direction, not to mention constructing a solvable problem aimed at minimizing the derived CRB. Consequently, we further define $\boldsymbol{\kappa =}\left[ {{\boldsymbol{\kappa }}_{1}^T},{{\boldsymbol{\kappa }}_{2}^T} \right]^T\text{, }{{\boldsymbol{\kappa }}_{1}}=\left[ {{d}_{o}},{{\phi }_{o}},\varphi  \right]^T\text{, }{{\boldsymbol{\kappa }}_{2}}=\left[ \mathbf{m}^T,\mathbf{n}^T \right]^T$, and make the following assumptions
\begin{equation}
{{\mathbf{f}}_{\boldsymbol{\kappa }}}=\left[ \begin{matrix}
   {{\mathbf{F}}_{{{\boldsymbol{\kappa }}_{1}}}} & \mathbf{F}_{{{\boldsymbol{\kappa }}_{2}},{{\boldsymbol{\kappa }}_{1}}}^{T}  \\
   {{\mathbf{F}}_{{{\boldsymbol{\kappa }}_{2}},{{\boldsymbol{\kappa }}_{1}}}} & {{\mathbf{F}}_{{{\boldsymbol{\kappa }}_{2}}}}  \\
\end{matrix} \right],
\end{equation}
\begin{equation}
{{\mathbf{f}}_{\boldsymbol{\kappa },g}}=\left[ 
   {\mathbf{f}}_{{{\boldsymbol{\kappa }}_{1}},g}^{T},\hspace{0.1cm} 
   {\mathbf{f}}_{{{\boldsymbol{\kappa }}_{2}},g}^{T} \right]^{T},
\end{equation}
\begin{equation}
\mathbf{J}\left( \boldsymbol{\kappa } \right)=\left[ \begin{matrix}
   \mathbf{J}\left( {{\boldsymbol{\kappa }}_{1}} \right) & \mathbf{J}^{T}{{\left( {{\boldsymbol{\kappa }}_{2}},{{\boldsymbol{\kappa }}_{1}} \right)}}  \\
   \mathbf{J}\left( {{\boldsymbol{\kappa }}_{2}},{{\boldsymbol{\kappa }}_{1}} \right) & \mathbf{J}\left( {{\boldsymbol{\kappa }}_{2}} \right)  \\
\end{matrix} \right],
\end{equation}

\begin{subequations}
\begin{align}
& \mathbf{CR}{{\mathbf{B}}_{\boldsymbol{\kappa }}}=\mathbf{J}^{-1}{{\left( \boldsymbol{\kappa } \right)}}={{\left( {{\mathbf{f}}_{\boldsymbol{\kappa }}}-{{\mathbf{f}}_{\boldsymbol{\kappa },g}}\mathbf{f}_{\boldsymbol{\kappa },g}^{T}/{f}_{g} \right)}^{-1}}, \\
\label{eq_31}
& \mathbf{CR}{{\mathbf{B}}_{{{\boldsymbol{\kappa }}_{1}}}}=\mathbf{J}^{-1}{{\left( {{\boldsymbol{\kappa }}_{1}} \right)}}={{\left( {{\mathbf{F}}_{{{\boldsymbol{\kappa }}_{1}}}}-{{\mathbf{f}}_{{{\boldsymbol{\kappa }}_{1}},g}}\mathbf{f}_{{{\boldsymbol{\kappa }}_{1}},g}^{T}/{f}_{g} \right)}^{-1}},
\end{align}
\end{subequations}
where the modified EFIM $\mathbf{J}\left( {{\boldsymbol{\kappa }}_{1}} \right)$ and $\mathbf{J}\left( {{\boldsymbol{\kappa }}_{2}} \right)$ compose the upper-left and lower-right matrix blocks of the original EFIM $\mathbf{J}\left( \boldsymbol{\kappa } \right)$. The elements of ${\mathbf{f}}_{\boldsymbol{\kappa }}$ and ${{\mathbf{f}}_{\boldsymbol{\kappa},g}}$, i.e. ${{\mathbf{F}}_{{{\boldsymbol{\kappa }}_{1}}}}$, ${{\mathbf{F}}_{{{\boldsymbol{\kappa }}_{1}},{{\boldsymbol{\kappa }}_{2}}}}$ and ${{\mathbf{f}}_{{{\boldsymbol{\kappa }}_{1}},g}}$, share similar definitions with the elements in (\ref{eq_ll100}). While $\mathbf{J}\left( {{\boldsymbol{\kappa }}_{1}} \right)$ is a shrunken version of $\mathbf{J}\left( \boldsymbol{\kappa } \right)$ with a size of $3 \times 3$, it provides a robust and closed-form CRB approximation for $\mathbf{J}\left( \boldsymbol{\kappa } \right)$ regarding the shared parameters ${{\boldsymbol{\kappa }}_{1}}$ in these two EFIMs, namely $CR{{B}_{{{\boldsymbol{\kappa }}_{1}}}}\left( {{d}_{o}} \right)\approx CR{{B}_{\boldsymbol{\kappa }}}\left( {{d}_{o}} \right)$, $CR{{B}_{{{\boldsymbol{\kappa }}_{1}}}}\left( {{\phi }_{o}} \right)\approx CR{{B}_{\boldsymbol{\kappa }}}\left( {{\phi }_{o}} \right)$ and $CR{{B}_{{{\boldsymbol{\kappa }}_{1}}}}\left( \varphi  \right)\approx CR{{B}_{\boldsymbol{\kappa }}}\left( \varphi  \right)$, which compose the main diagonal of matrix $CR{{B}_{{{\boldsymbol{\kappa }}_{1}}}}$. \textcolor{black}{Such an approximation becomes accurate when the ET is sufficiently distant from the BS.}

From another aspect, the contour parameter ${{\boldsymbol{\kappa }}_{2}}$ is actually invariant for one specific ET. We can exclude ${{\boldsymbol{\kappa }}_{2}}$ from $\boldsymbol{\kappa }$ to be estimated when the estimating accuracy is acceptable after several observations, or when ${{\boldsymbol{\kappa }}_{2}}$ is already known by the BS. As such, the whole EFIM $\mathbf{J}\left( \boldsymbol{\kappa } \right)$ naturally degrades to the modified EFIM $\mathbf{J}\left( {{\boldsymbol{\kappa }}_{1}} \right)$. In the following, we shall use $CRB$ to replace $CR{{B}_{{{\boldsymbol{\kappa }}_{1}}}}$ for simplicity. Finally, with prior knowledge of the RCS distribution and the exact contour of the target, we have the following theorem:

\begin{theorem}
The CRBs on ${{d}_{o}}$, ${{\phi }_{o}}$ and $\varphi $ can be expressed as $(\ref{CRB d_o})-(\ref{CRB varphi})$, respectively, where ${{Z}_{1,k}}={{\pi }^{2}}\left( N_{r}^{2}-1 \right){{\cos }^{2}}\phi_{k} /12$, ${{Z}_{2}}={{\left( 4\pi B/c \right)}^{2}}$, \textcolor{black}{${{X}_{k}}={{- \boldsymbol{\nu} _{k}^{T}\mathbf{m}\cos \left( {{\phi }_{o}}+\varphi  \right)+\boldsymbol{\varsigma} _{k}^{T}\mathbf{n}\sin \left( {{\phi }_{o}}+\varphi  \right) }}$}, ${{\mathbf{a}}_{k}}$ is the abbreviation for $\mathbf{a}\left( {{\phi }_{k}} \right)$, ${{\boldsymbol{\nu} }_{k}}={{\left[ \cos \left( {{u}_{k}} \right),...,\cos \left( Q{{u}_{k}} \right) \right]}^{T}}$ and ${{\boldsymbol{\varsigma} }_{k}}={{\left[ \sin \left( {{u}_{k}} \right),...,\sin \left( Q{{u}_{k}} \right) \right]}^{T}}$ are the cosine and sine harmonics, respectively.
\end{theorem}

\begin{proof}
Please see Appendix I.
\end{proof}

\begin{remark}
Observing the structures of the CRBs in $(\ref{CRB d_o})-(\ref{CRB varphi})$, we can find some interesting facts: 1) While the CRB on ${{\phi }_{o}}$ has merely one summation operation over the contour subsections, the CRBs on ${{d}_{o}}$ and $\varphi $ are typically more sophisticated with multiple summations; 2) The CRBs on angle parameters are closely related, where $CRB\left( \varphi  \right)$ can be expressed as the sum of $CRB\left( {{\phi }_{o}} \right)$ and a term similar to $CRB\left( {{d}_{o}} \right)$. This also indicates that target orientation $\varphi$ is harder to be estimated than the direction $\phi_{o}$; 3) $CRB\left( {{\phi }_{o}} \right)$ for ET can be viewed as an extended version of $CRB\left( {{\phi }_{o}} \right)$ for PT in \cite{Liu21TSP}, despite some coefficient difference.

Since the direction is generally recognized as a crucial parameter in ISAC systems for effective beamforming, in the subsequent sections we shall focus on the minimization of $CRB\left( {{\phi }_{o}} \right)$, setting aside the discussion of other CRBs.
\end{remark}

\subsection{CRB for Point Target}
In the radar literature, a target located far away from the radar is generally considered as a PT. This is a special case of ET with zero contour and zero extend. In this subsection, we follow the similar derivation procedures in the previous subsection to obtain the CRBs for PT and then establish their connections. Following the notations in $(\ref{eq_ll8})-(\ref{eq14})$, the echo signals and the received sensing signals are
\begin{equation}
{{\mathbf{s}}_{p}}\left( t \right)=g{{\alpha}_{p}}\mathbf{b}\left( {{\phi }_{o}} \right){{\mathbf{a}}^{H}}\left( {{\phi }_{o}} \right)\mathbf{x}\left( t-{2{{d}_{o}}}/{c_0} \right),
\end{equation}
\begin{equation}
\mathbf{y}_{s,p}=\mathbf{s}_{p}\left( t \right)+\mathbf{z}_s (t),
\end{equation}
where the subscript $p$ refers to the PT case, ${{\alpha}_{p}}$ describes the target RCS. We can observe that sensing a PT is typically less difficult compared with sensing an ET, since the number of unknown parameters shrinks from $2Q+3$ to $2$, and we only need to estimate ${{d}_{o}}$ and ${{\phi }_{o}}$ for a PT. Applying similar derivation process of the CRBs for ET, the CRBs for PT are
\begin{equation}
\label{eq31}
CR{{B}_{p}}\left( {{d}_{o}} \right)={{\left( {2g_{p}^{2}{{N}_{r}}{{Z}_{2}}}/{\sigma _{s}^{2}} \right)}^{-1}}{{\left( \mathbf{a}_{o}^{H}{{\mathbf{R}}_{x}}{{\mathbf{a}}_{o}} \right)}^{-1}},
\end{equation}
\begin{align}
\label{eq32}
&CRB_{p}\left( {{\phi }_{o}} \right)={\left({2{{g}_{p}^{2}}{{N}_{r}}t_s}/{\sigma _{s}^{2}} \right)}^{-1}\nonumber\\
&\left[ {{Z}_{1,o}}\mathbf{a}_{o}^{H}{{\mathbf{R}}_{x}}{{\mathbf{a}}_{o}}+\mathbf{\dot{a}}_{o}^{H}{{\mathbf{R}}_{x}}{{{\mathbf{\dot{a}}}}_{o}}-\frac{{{\left( \mathbf{\dot{a}}_{o}^{H}{{\mathbf{R}}_{x}}{{\mathbf{a}}_{o}}+\mathbf{a}_{o}^{H}{{\mathbf{R}}_{x}}{{{\mathbf{\dot{a}}}}_{o}} \right)}^{2}}}{4\mathbf{a}_{o}^{H}{{\mathbf{R}}_{x}}{{\mathbf{a}}_{o}}} \right]^{-1},
\end{align}
where ${{\mathbf{a}}_{o}}$ is the abbreviation of $\mathbf{a}\left( {{\phi }_{o}} \right)$, the auxiliary variables ${{Z}_{1,o}}$ and ${{Z}_{2}}$ share the same definition with $(\ref{CRB d_o})-(\ref{CRB phi_o})$. Comparing $(\ref{CRB d_o})-(\ref{CRB phi_o})$ with $(\ref{eq31})-(\ref{eq32})$, if we eliminate the terms related to the contour intermediate variable $X_k$ and normalize the length variable ${{l}_{k}}$, it can be observed that CRBs of PT are equivalent to the modified CRBs of ET with a single reflection point, revealing the connection between PT and ET. 

\section{Joint Beamforming Design}

In this section, we focus on the joint beamforming design for the ISAC system. We first present the general ISAC beamforming optimization problem in Section IV-A, which aims at minimizing the CRB for direction estimation under communication and beam coverage constraints. To solve this problem, a standard SDR beamforming algorithm and an efficient ZF beamforming algorithm are presented in Section IV-B and Section IV-C, respectively.

\subsection{General Beamforming Optimization Problem}
Our goal is to optimize the beamformers for minimizing CRB of the direction estimation for ET in $(\ref{CRB phi_o})$, under some practical constraints. The problem can be formulated as
\begin{align}
(\mathcal{P}1):\hspace{1pt}&\underset{\left\{ {{\mathbf{w}}_{n}} \right\}_{n=1}^{N_c}}{\mathop{\min }}\,CRB\left( {{\phi }_{o}} \right) \label{Problem P1}\\
\text{s.t. }&\mathrm{tr}\left( {{\mathbf{R}}_{x}} \right)\le {{P}_{t}},\mathbf{R}_{x}=\sum\nolimits_{n=1}^{N_c}\mathbf{R}_{n}\tag{\ref{Problem P1}{a}}\label{Problem P1a}\\
&\left( 1+{{\Gamma }^{-1}} \right)\mathbf{h}_{n}^{H}{{\mathbf{R}}_{n}}{{\mathbf{h}}_{n}}\geq \mathbf{h}_{n}^{H}{{\mathbf{R}}_{x}}{{\mathbf{h}}_{n}}+\sigma _{n}^{2},\forall n,\tag{\ref{Problem P1}{b}}\label{Problem P1b}\\
&2\min\limits_{1 \leq k \leq K} \left( \mathbf{a}_{k}^{H}{{\mathbf{R}}_{x}}{{\mathbf{a}}_{k}} \right)-\max\limits_{1 \leq k \leq K} \left( \mathbf{a}_{k}^{H}{{\mathbf{R}}_{x}}{{\mathbf{a}}_{k}} \right)\geq 0,\tag{\ref{Problem P1}{c}}\label{Problem P1c}
\end{align}
here, constraint $(\mathrm{\ref{Problem P1a}})$ is to ensure a total transmit power constraint at the BS, and constraint $(\mathrm{\ref{Problem P1b}})$ is to guarantee the minimum SINR requirement for each CU, where $\Gamma$ is the SINR threshold for all CUs and $\mathbf{R}_{n}=\mathbf{w}_{n}\mathbf{w}_{n}^{H}$ is the covariance matrix for the $n$-th CU. The last constraint $(\mathrm{\ref{Problem P1c}})$ is a beam coverage constraint for ET sensing. The reason for adding $(\mathrm{\ref{Problem P1c}})$ is as follows. 

To facilitate the sensing of ET instead of PT, we need to ensure that every element on the ET LoS contour receives sufficient energy and reflects off valid echo signals, which is vital for further ET parameter estimation. We use the beampattern term $\mathbf{a}_{k}^{H}{{\mathbf{R}}_{x}}{{\mathbf{a}}_{k}}$ to describe the energy received by the elements along the $k$-th LoS contour subsection. Thus, $(\mathrm{\ref{Problem P1c}})$ serves as a special 3-dB beam coverage constraint, where the maximum received energy of each LoS contour element should cater to the corresponding minimum value. Note that in $(\ref{eq_ll8})$, we discretize the LoS contour curve into $K$ subsections, and $(\mathrm{\ref{Problem P1c}})$ can be expressed in a max-min discrete form in $\mathcal{P}1$.

Recalling the expression in $(\ref{CRB phi_o})$, it is clear that the objective function $CRB\left( {{\phi }_{o}} \right)$ is non-convex owing to its fractional structure. Here we introduce an auxiliary variable $t$ to reshape $CRB\left( {{\phi }_{o}} \right)$ in $\mathcal{P}1$. We assume that
\begin{align}
\label{eq_add1}
&\textcolor{black}{\sum\limits_{k=1}^{K}l_k}\left({{Z}_{1,k}}\mathbf{a}_{k}^{H}{{\mathbf{R}}_{x}}{{\mathbf{a}}_{k}}+\mathbf{\dot{a}}_{k}^{H}{{\mathbf{R}}_{x}}{{{\mathbf{\dot{a}}}}_{k}}\right)-\nonumber\\
&\hspace{1.3cm}{\left[\textcolor{black}{\sum\limits_{k=1}^{K}l_k}{\mathcal{R}\left( \mathbf{\dot{a}}_{k}^{H}{{\mathbf{R}}_{x}}{{\mathbf{a}}_{k}}\right)}\right]^{2}}/{\textcolor{black}{\sum\limits_{k=1}^{K}l_k}\mathbf{a}_{k}^{H}{{\mathbf{R}}_{x}}{{\mathbf{a}}_{k}}} \geq t.
\end{align}

Utilizing the Schur complement condition \cite{Liu21TSP}, we additionally introduce an auxiliary second-order matrix $\mathbf{P}$, where $(\ref{eq_add1})$ can be transformed into a semidefinte form as
\begin{subequations}
\begin{align}
\label{P_k}
&\textcolor{black}{\mathbf{P}= \sum\limits_{k=1}^{K}l_k\mathbf{P}_k \succeq 0,}\\
&\mathbf{P}_k = \left[ \begin{matrix}
   \left({{Z}_{1,k}}\mathbf{a}_{k}^{H}{{\mathbf{R}}_{x}}{{\mathbf{a}}_{k}}+\mathbf{\dot{a}}_{k}^{H}{{\mathbf{R}}_{x}}{{{\mathbf{\dot{a}}}}_{k}}\right)-\frac{t}{Kl_k} & {\mathcal{R}\left( \mathbf{\dot{a}}_{k}^{H}{{\mathbf{R}}_{x}}{{\mathbf{a}}_{k}}\right)}  \\
  {\mathcal{R}\left( \mathbf{\dot{a}}_{k}^{H}{{\mathbf{R}}_{x}}{{\mathbf{a}}_{k}}\right)} & \mathbf{a}_{k}^{H}{{\mathbf{R}}_{x}}{{\mathbf{a}}_{k}}
\end{matrix} \right].
\end{align}
\end{subequations}

Finally, we get the equivalent form of problem $\mathrm{\mathcal{P}1}$ as
\begin{align}
(\mathcal{P}2):\hspace{1pt}& \underset{\left\{ {{\mathbf{w}}_{n}} \right\}_{n=1}^{N_c},{t}}{\mathop{\min }}\,-\textcolor{black}{t}\label{Problem P2}\\ 
&\text{s.t. } \mathbf{P}\succeq 0, \tag{\ref{Problem P2}{a}}\label{Problem P2a}\\ 
&\hspace*{0.53cm}(\mathrm{\ref{Problem P1a}}),(\mathrm{\ref{Problem P1b}}),(\mathrm{\ref{Problem P1c}})\nonumber,
\end{align}
where $\mathbf{P}$ in problem $\mathcal{P}2$ is consistent with that in $(\mathrm{\ref{P_k}})$.

\subsection{Joint Transmit Beamforming Design via SDR}
We observe that problem $\mathcal{P}2$ is still non-convex owing to the quadratic terms in the constraints, for instance ${{\mathbf{R}}_{n}}={{\mathbf{w}}_{n}}\mathbf{w}_{n}^{H}$. To formulate a convex problem, one common practice is to employ the SDR technique, replacing the original variable ${{\mathbf{w}}_{n}}$ by ${{\mathbf{R}}_{n}}={{\mathbf{w}}_{n}}\mathbf{w}_{n}^{H}$ with rank-one constraint $\text{rank}\left( {{\mathbf{R}}_{n}} \right)=1$. Omitting this rank-one constraint leads to the semidefinite programming (SDP) problem as follows
\begin{align}
(\mathcal{P}3):\hspace{1pt}&\underset{\left\{ {{\mathbf{R}}_{n}} \right\}_{n=1}^{N_c},t }{\mathop{\min }}\,-t\label{Problem P3}\\ 
 &\text{s.t. } \mathbf{R}_n \succeq 0, \forall n, \tag{\ref{Problem P3}{a}}\label{Problem P3a}\\
 &\hspace*{0.53cm}(\mathrm{\ref{Problem P1a}}),(\mathrm{\ref{Problem P1b}}),(\mathrm{\ref{Problem P1c}}),(\mathrm{\ref{Problem P2a}})\nonumber.
\end{align}

It can be noted that the relaxed problem $\mathcal{P}3$ is a standard SDP problem whose global optimum can be efficiently obtained by convex optimization toolboxes \cite{ref29,ref30}. While the solution for $\mathcal{P}3$ usually holds optimal for the initial non-convex problem $\mathcal{P}2$, we can not guarantee $\mathcal{P}3$ to be a tight relaxation of $\mathcal{P}2$ for all circumstances owing to the complicated constraints. In other words, the optimal solutions $\left\{ {{\mathbf{R}}_{n}} \right\}_{n=1}^{N_c}$ for $\mathcal{P}3$ are not necessarily all rank-one. Consequently, we present a rank-one solution extracting algorithm to obtain the valid beamformers, as outlined in Algorithm \ref{alg1}.

\begin{algorithm}
\caption{Extract rank-one beamformers from $\mathcal{P}$2.}
\begin{algorithmic}
\label{alg1}
\STATE 
\STATE {\textsc{Input:}}
\STATE \hspace{0.5cm}$ \text{$\mathcal{P}$2 solutions } \{\mathbf{R}_n\}_{n=1}^{N_c}  $
\STATE {\textsc{Output:}}
\STATE \hspace{0.5cm}$ \text{Rank-one transmit beamformers } \{\mathbf{\tilde{w}}_{n}\}_{n=1}^{N_c}  $
\STATE {\textsc{Steps:}}
\begin{enumerate}
    \item {Randomly extract $\mathbf{\tilde{u}}_{n}\in \mathrm{span} \left( \mathbf{R}_n \right)$}
    \item {Execute \begin{flalign}
        &\ \mathbf{\tilde{\eta}}=\left[\gamma_{1}\sigma _{n}^{2},...\gamma_{N_c}\sigma _{n}^{2}\right]^{T}, \nonumber \\
        \nonumber&\left[\mathbf{\tilde{F}}\right]_{i,n}=
        \begin{cases}
        \begin{aligned}
        &\mathbf{\tilde{u}}_{i}^{H}\mathbf{h}_{i}^{H}\mathbf{h}_{i}\mathbf{\tilde{u}}_{i}, && i=n, \\
        &\Gamma \mathbf{\tilde{u}}_{i}^{H}\mathbf{h}_{n}^{H}\mathbf{h}_{n}\mathbf{\tilde{u}}_{i}, && i\ne n,
        \end{aligned}
        \end{cases} &\\ 
        &\ \mathbf{\tilde{q}}=\mathbf{\tilde{F}}^{-1}\mathbf{\tilde{\eta}}=\left[\mathbf{\tilde{q}}_{1},...,\mathbf{\tilde{q}}_{N_c}\right]^{T}, & \nonumber \\
        &\ \mathbf{\tilde{w}}_{n}=\sqrt{\mathbf{\tilde{q}}}_{n}\mathbf{\tilde{u}}_{n}, n=1,...,N_c. & \nonumber
        \end{flalign}}
\end{enumerate}
\end{algorithmic}
\end{algorithm}

\begin{theorem}
If $\left\{ {{\mathbf{R}}_{n}} \right\}_{n=1}^{N_c}$ is the optimal solution to $\mathcal{P}3$, then Algorithm \ref{alg1} gives a solution $\left\{\tilde{\mathbf{w}}_c\right\}_{n=1}^{N_c}$ to $\mathcal{P}2$ without violating constraints $(\mathrm{\ref{Problem P1a}})$ and $(\mathrm{\ref{Problem P1b}})$.
\end{theorem}

\begin{proof}
See Chapter 18 in \cite{ref31}.
\end{proof}

Note that one single solution $\{{\mathbf{\tilde{w}}}_{n}\}_{n=1}^{N_c}$ from Algorithm \ref{alg1} is not necessarily guaranteed to satisfy the beam coverage constraint $(\mathrm{\ref{Problem P1c}})$. Thus, we can execute Algorithm \ref{alg1} multiple times till constraint $(\mathrm{\ref{Problem P1c}})$ is met.

\subsection{Joint Transmit Beamforming Design via ZF}
Solving problem $\mathcal{P}3$ in previous section suffers from high complexity and inevitable randomness due to Algorithm \ref{alg1}, which motivates us to explore a more efficient and deterministic beamforming algorithm. In this subsection, we utilize ZF to design transmit beamformers. ZF beamforming is a well-known suboptimal beamforming approach \cite{ref34} which provides a tradeoff between complexity and system performance. The basic ZF conditions can be formulated as ${{\left[ \mathbf{HW} \right]}_{i,n}}=0,\text{ }i\ne n$ and ${{\left[ \mathbf{HW} \right]}_{i,n}}={{p}_{n}},i=n$. The matrix notation is expressed as
\begin{equation}
\label{eq38}
\mathbf{HW}=\text{diag}\left\{ \sqrt{\mathbf{p}} \right\},
\end{equation}
where $\mathbf{p=}\left[ {{p}_{1}},...,{{p}_{N_c}} \right]^T$ is the power allocation vector. It can be observed that ZF naturally decouples the multiuser channel into $N_c$ independent subchannels, transforming the beamforming design to a power allocation problem. We assume that $\mathbf{H}$ is a full row-rank matrix. The generalized inverse matrix of $\mathbf{H}$ is
\begin{equation}
\label{eq39}
{{\mathbf{H}}^{-}}={{\mathbf{H}}^{\dagger }}+{{\mathbf{P}}_{\bot }}\mathbf{U},
\end{equation}
where ${{\mathbf{H}}^{\dagger }}={{\mathbf{H}}^{H}}{{\left( \mathbf{H}{{\mathbf{H}}^{H}} \right)}^{-1}}=\left[ \mathbf{h}_{1}^{^{\dagger }},...,\mathbf{h}_{N_c}^{^{\dagger }} \right]$ is the pseudo-inverse of $\mathbf{H}$, ${{\mathbf{P}}_{\bot }}={{\mathbf{I}}_{{{N}_{t}}}}-{{\mathbf{H}}^{\dagger }}\mathbf{H}$ orthogonally projects any vector onto the null space of $\mathbf{H}$, $\mathbf{U}=\left[\mathbf{u}_{1},...,\mathbf{u}_{N_c}\right]$ is an arbitrary matrix yet to be designed. Combining $(\ref{eq38})$ and $(\ref{eq39})$, a universal form for ZF beamforming is expressed as
\begin{align}
&\mathbf{W}={{\mathbf{H}}^{-}}\mathrm{diag}\left\{ \sqrt{\mathbf{p}} \right\}=\left( {{\mathbf{H}}^{\dagger }}+{{\mathbf{P}}_{\bot }}\mathbf{U} \right)\mathrm{diag}\left\{ \sqrt{\mathbf{p}} \right\}, \\
\label{eq149}
&\hspace{2cm}{{\mathbf{w}}_{n}}=\sqrt{{{p}_{n}}}\left( \mathbf{h}_{n}^{\dagger }+{{\mathbf{P}}_{\bot }}{{\mathbf{u}}_{n}} \right).
\end{align}

From $(\ref{eq149})$, the original beamforming design of $\mathbf{W}$ is transformed to optimizing the power allocation vector $\mathbf{p}$ and the arbitrary matrix $\mathbf{U}$. To guarantee robust target estimation performance in radar systems, a common practice is to steer the probing signals straightly towards the direction of target. Then, we decompose each vector $\mathbf{u}_c$ in $\mathbf{U}$ into a linear combination of target contour subsection steering vectors, namely

\begin{equation}
{{\mathbf{u}}_{n}}=\sum\nolimits_{k=1}^{K}{{{\eta }_{n,k}}\mathbf{a}_k},
\end{equation}
here the subscripts $n$, $k$ represent that the steering vector of the $k$-th contour subsection direction is assigned to the $n$-th CU’s beamformer, ${{\eta }_{n,k}}$ is the steering vector coefficient. The optimization problem is thus presented as
\begin{align}
(\mathcal{P}4):\hspace{1pt} & \underset{\left\{ {{\eta }_{n,k}} \right\}_{n=1,k=1}^{N_c,K},\left\{ {{p}_{n}} \right\}_{n=1}^{N_c},t}{\mathop{\min }}\,-t \label{Problem P4}\\ 
  &\text{s.t. } p_c\ge \Gamma \sigma _{n}^{2},\forall n, \tag{\ref{Problem P4}{a}}\label{Problem P4a}\nonumber\\ 
 &\hspace*{0.53cm} {{\mathbf{w}}_{n}}=\sqrt{{{p}_{n}}}\left[ \mathbf{h}_{n}^{\dagger }+{{\mathbf{P}}_{\bot }}\sum\nolimits_{k=1}^{K}{{{\eta}_{n,k}}\mathbf{a}_k} \right],\forall n\tag{\ref{Problem P4}{b}}\label{Problem P4b}\nonumber\\
 &\hspace*{0.53cm}(\mathrm{\ref{Problem P1a}}),(\mathrm{\ref{Problem P1c}}),(\mathrm{\ref{Problem P2a}}).\nonumber
\end{align}

The optimization problem $\mathcal{P}4$ is non-convex due to the multiplicative coupling between variables, i.e. ${{\mathbf{u}}_{n}}\mathbf{u}_{n}^{H}=\sum\nolimits_{{{k}_{1}}=1}^{K}{\sum\nolimits_{{{k}_{2}}=1}^{K}{{{\eta }_{n,{{k}_{1}}}}}{{\eta }_{n,{{k}_{2}}}}{{\mathbf{a}}_{{{k}_{1}}}}\mathbf{a}_{{{k}_{2}}}^{H}}$ has a non-convex term of ${{\eta }_{n,{{k}_{1}}}}{{\eta }_{n,{{k}_{2}}}}$. To avoid the non-convex coupling, the column vectors in $\mathbf{U}$ is further modified to align with merely one contour subsection direction, namely
\begin{equation}
{{\mathbf{u}}_{n}}={{\eta }_{n}}{{\mathbf{a}}_{{{k}_{n}}}},
\end{equation}
here ${{k}_{n}}$ represents the subsection index aligned to the $n$-th CU which is selected from the contour subsection set $\mathcal{K}=\left\{ 1,...,K \right\}$, ${{\eta }_{n}}$ is the steering vector coefficient. The beamformer can be obtained as
\begin{equation}
\label{eq45}
{{\mathbf{w}}_{n}}=\sqrt{{{p}_{n}}}\left( \mathbf{h}_{n}^{\dagger }+{{\mathbf{P}}_{\bot }}{{\mathbf{u}}_{n}} \right)=\sqrt{{{p}_{n}}}\left( \mathbf{h}_{n}^{\dagger }+{{\mathbf{P}}_{\bot }}{{\eta }_{n}}{{\mathbf{a}}_{{{k}_{n}}}} \right).
\end{equation}

From $(\ref{eq45})$, we can observe that there still exists non-convex multiplicative coupling between variables ${{p}_{n}}$ and ${{\eta }_{n}}$. Consequently, we conditionally omit ${{p}_{n}}$ and reconstruct an equivalent beamforming vector ${{\mathbf{w}}_{n}}$ as follows
\begin{equation}
\label{w_n}
{{\mathbf{w}}_{n}}=\sqrt{{{p}_{n}}}\mathbf{h}_{n}^{\dagger }+{{\mathbf{P}}_{\bot }}{{\mathbf{u}}_{n}}=\sqrt{{{p}_{n}}}\mathbf{h}_{n}^{\dagger }+{{\mathbf{P}}_{\bot }}{{\eta }_{n}}{{\mathbf{a}}_{{{k}_{n}}}}.
\end{equation}

The equivalence between the non-omitted and modified beamformer in $(\ref{eq45})$ and $(\ref{w_n})$ is straightforward. Given an optimal solution ${{p}_{n}}$ and $\hat{\eta}_n$ regarding the non-omitted beamforming vector $(\ref{eq45})$, we can simply obtain the same beamformer by defining \textcolor{black}{${\eta}_n=\hat{\eta}_n/{\sqrt{{p}_{n}}}$} in $(\ref{w_n})$, and vice versa. The final ZF beamforming problem can be formulated as follows
\begin{align}
(\mathcal{P}5):\hspace{1pt} & \underset{\left\{ {{\eta }_{n}},{{p}_{n}} \right\}_{n=1}^{N_c},t}{\mathop{\min }}\,-t\label{Problem P5}\\ 
  &\text{s.t. } {{\mathbf{w}}_{n}}=\sqrt{{{p}_{n}}}\mathbf{h}_{n}^{\dagger }+{{\mathbf{P}}_{\bot }}{{\eta }_{n}}{{\mathbf{a}}_{{{k}_{n}}}},\forall n,\tag{\ref{Problem P5}{a}}\label{Problem P5a}\\
  &\hspace*{0.53cm}(\mathrm{\ref{Problem P1a}}),(\mathrm{\ref{Problem P1c}}),(\mathrm{\ref{Problem P2a}}),(\mathrm{\ref{Problem P4a}}).\nonumber
\end{align}

We note that the exact direction ${{\phi }_{{{k}_{n}}}}\in \left\{ {{\phi }_{1}},...,{{\phi }_{K}} \right\}$ for all beamformers should be determined prior to solving the ZF problem, namely we need to select $N_c$ out of totally $K$ directions in advance. While it would be very efficient to utilize one feasible and robust direction combination $\left\{ {{\phi }_{{{k}_{1}}}},...,{{\phi }_{{{k}_{N_c}}}} \right\}$ by experience to solve $\mathcal{P}5$, more generally we need to try through at total $\mathrm{C}_{K}^{N_c}$ direction sets in $\mathcal{P}5$ and find the optimal one for beamforming design. For each set, problem $\mathcal{P}5$ can be efficiently solved by existing toolboxes \cite{ref29,ref30}. The optimal solution is defined as the beamforming matrix with minimum CRB as follows
\begin{equation}
\left\{\mathbf{w}_{n}^{opt}\right\}_{n=1}^{N_c}=\min\limits_{1\leq e \leq \text{C}_{K}^{N_c} } {CRB}\left( \left\{ {\mathbf{w}_{n,e}} \right\}_{n=1}^{N_c} \right),
\label{eq48}
\end{equation}
where the subscript $e$ refers to the direction set index.

\subsection{Complexity Analysis}
The primary computation burdens in SDR and ZF beamforming originate from solving the SDP problems $\mathcal{P}3$ and $\mathcal{P}5$. With a given solution accuracy $\varepsilon $, the worst case complexity to solve the SDR beamforming problem $\mathcal{P}3$ with the primal-dual interior-point method is $\mathcal{O}\left( N_c^{6.5}N_t^{6.5}\log \left( 1/\varepsilon  \right) \right)$, whereas the complexity to solve ZF beamforming problem $\mathcal{P}5$ is $\mathcal{O}\left( N_c^{6.5}\mathrm{C}_{K}^{N_c}\log \left( 1/\varepsilon  \right) \right)$. Briefly omitting the $\mathrm{C}_{K}^{N_c}$ term caused by the iteration of direction sets, the worst-case complex flops for ZF beamforming is greatly lowered by a factor of $N_{t}^{6.5}$ compared to SDR beamforming. Since the flops caused by constraints is negligible in both SDR and ZF algorithms, the significant complexity reduction is mostly attributed to the sharp decrease of variable element numbers. Specifically, problem $\mathcal{P}3$ has $N_c$ matrix variables and one scalar variable to be optimized with totally $N_{c}N_{t}^2+1$ elements, whereas problem $\mathcal{P}5$ only has $2N_c+1$ scalar variables to be optimized with totally $2N_c+1$ elements.

\section{Numerical Results}
In this section, we provide numerical results to illustrate the derived CRBs and evaluate the performance of the proposed beamforming design. If not specifically indicated, we consider an ISAC BS equipped with ${{N}_{t}}=16$ transmit antennas and ${{N}_{r}}=16$ receive antennas. The transmit power is ${{P}_{t}}=0$ dBW, the noise powers of communication and sensing are $\sigma _{n}^{2}=\sigma _{s}^{2}=-80$ dBm. There exist $N_c=4$ downlink CUs located at $\boldsymbol{\phi}=\left[ -{{60}^\circ},-{{35}^\circ},{{35}^\circ},{{60}^\circ} \right]$. The path loss of each CU is 100 dB. The LoS signal component of each CU occupies $\mathrm{90\%}$ of the received signal strength, and there exist $L=6$ paths between the BS and each CU. The SINR threshold is set as $\Gamma =10$ dB. An ET with vehicle shape is assumed to be located ${{d}_{o}}=27$ m away from the BS with a direction of ${{\phi }_{o}}={0}^\circ$ and an orientation of $\varphi ={{0}^\circ}$. The whole LoS contour is divided into $K=8$ disjoint subsections. The length and width of the target contour are approximately 5 m and 2 m, respectively, and the TFS contour is parameterized by $Q = 8$ harmonics. The observation period is set as $t_s = 1\text{ s}$.

\begin{figure}[!t]
\centering
\includegraphics[width=3.5in]{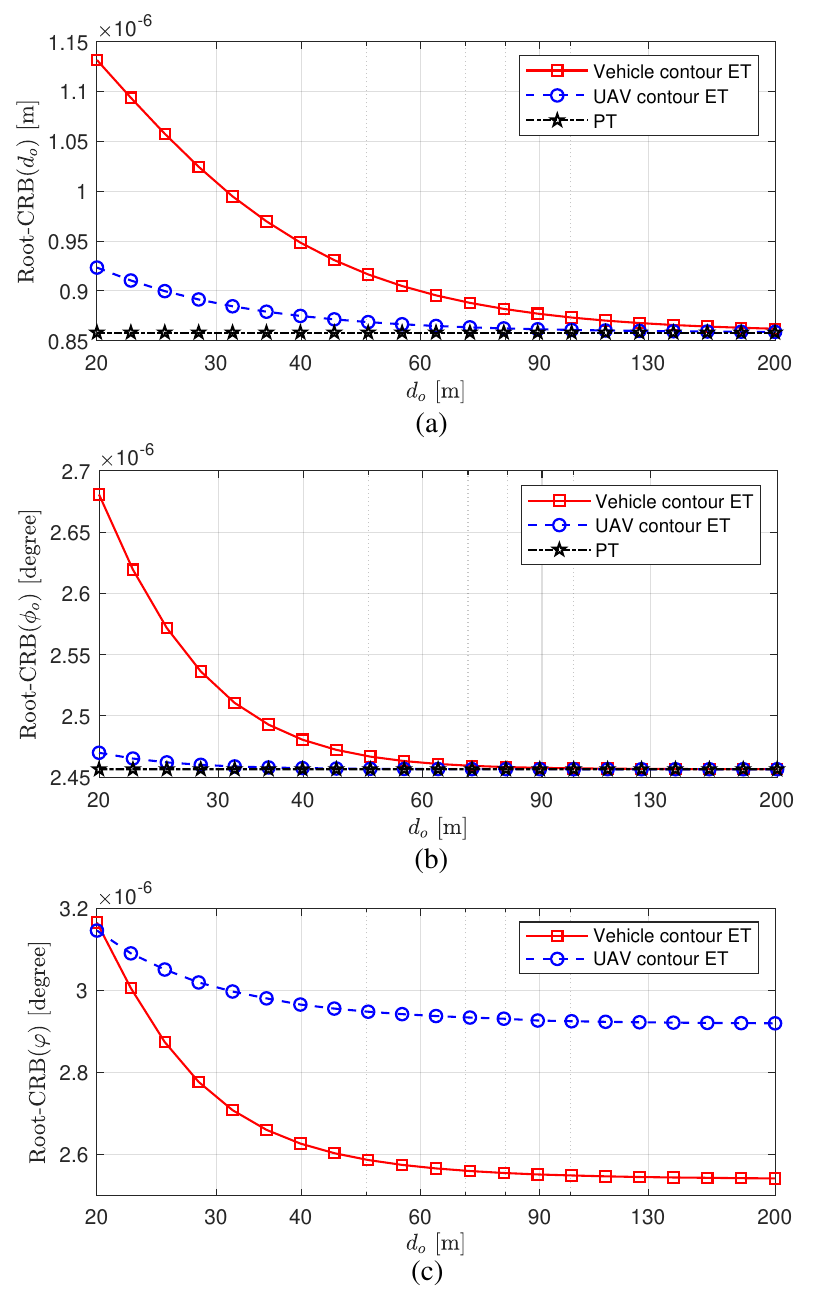}
\caption{\textcolor{black}{CRBs of target parameters versus distance. (a) range; (b) direction; (c) orientation.}}
\label{Fig3}
\end{figure}

For comparsion, we define the SDR beamforming design proposed in Section IV as the `CRB-min Design'. Two beampattern-approaching beamforming designs proposed in \cite{Hua23TVT} and \cite{Liu20TSP} are selected as benchmarks, referred as `Average Design' and `Average-Null Design' respectively. Average Design in \cite{Hua23TVT} tends to allocate equal energy towards areas of interest, while, on the basis of \cite{Hua23TVT}, Average-Null Design in \cite{Liu20TSP} additionally prefers the energy transmitted elsewhere to be zero. We employ a main beam with a half-power beamwidth of $\Delta ={10}^\circ$ in the simulation. The global direction grids are uniformly sampled among $\left[ -{90}^\circ,{90}^\circ \right]$ with an interval of $1^\circ$.

\subsection{Numerical analysis of CRB}
In this subsection, we report the numerical analysis of CRB of ET parameters. Two different ET contours, namely, vehicle and unmanned aerial vehicle (UAV), are chosen for the numerical computation of range, direction, and orientation CRBs. The contour parameters for the vehicle shape are $\mathbf{m}=$[2.05, -0.002, 0.5, 0, 0.056, 0.001, -0.125, 0.003]$^T$, $\mathbf{n}=$[1.24, -0.001, 0.335,-0.001, 0.124, -0.001, 0.018, 0]$^T$, while those for the UAV shape are $\mathbf{m}=$[0.797, 0, -0.153, 0, -0.272, 0, -0.12, 0, 0.045]$^T$, $\mathbf{n}=$[0.797, 0, 0.153, 0, -0.272, 0, 0.12, 0, 0.045]$^T$. To facilitate the comparison between ET and PT, we set the same received energy for all considered targets by normalizing the ET contour lengths with $\sum\nolimits_{k=1}^{K}l_k=1$. From $(\ref{CRB d_o})-(\ref{CRB varphi})$, we learn that CRB scales to the fourth power of range and changes sensitively with the range variations. For ease of analysis, we set the radar SNR ${{\gamma }_{s}}={{N}_{r}}{{P}_{t}}/\left( d_{o}^{4}\sigma _{s}^{2} \right)$ a constant regardless of the distance between the target and BS. The ETs and PT are assumed to move along $+y$ axis from $[0,20]$ to $[0,200]$. All the points are averaged over 2000 beamforming matrix realizations.

Fig. \ref{Fig3}(a)$-$Fig. \ref{Fig3}(c) illustrate the CRBs for range $d_o$, direction $\phi_o$, and orientation $\varphi$ of two different contours, respectively, as a function of the distance between the target and the BS. The CRBs of PT are also plotted for comparison. From Fig. 3(a), it is seen that while the range CRB of PT is actually invariant of distances, the range CRBs of ETs have greatly diversified response towards target distances. \textcolor{black}{Particularly, the range CRB of the UAV-shaped ET first decreases sharply, then converges to a constant as the distance increases. A similar and more obvious decreasing trend can be observed for the range CRB of the vehicle-shaped ET throughout the target movement. Note that the range CRBs of ETs converge to the range CRB of PT at large distances, which verifies the equivalence between the CRBs of faraway ET and PT.}

From Fig. \ref{Fig3}(b), it is seen that the direction CRBs of the two different ETs share similar decreasing trends with the range CRBs when the distance increases, and they both converge to the direction CRB of PT at large distance, say $d_o\geq70\text{m}$. By comparing Fig. \ref{Fig3}(b) and Fig. \ref{Fig3}(c), we also find that the orientation CRB is constantly larger than the corresponding direction CRB, which corresponds to the CRB analysis concluded in Section III. \textcolor{black}{The obtained results are substantially different for the SIMO radar sensing scenario \cite{Garcia22TSP}, for instance the CRBs of ET with a known contour almost share the same value with those of PT.}

 \subsection{Comparison Between Proposed Beamformer and Benchmark Beamformers}
 In this section, we compare different beamforming schemes from three aspects, i.e. the estimation accuracy of unknown parameter, the beampattern, and the sensing CRB.

\begin{figure}[!t]
\centering
\includegraphics[width=3.5in]{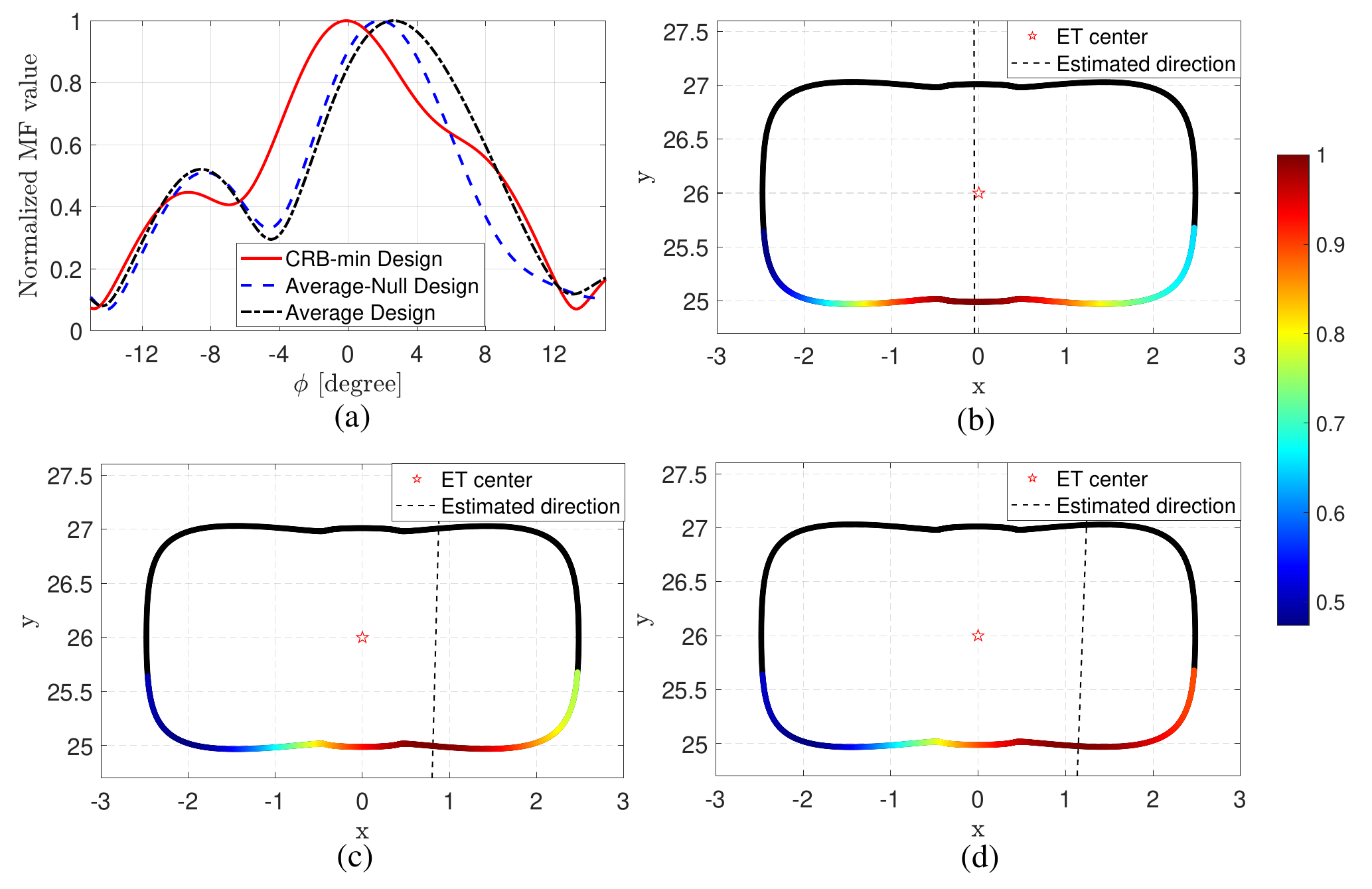}
\caption{The MF-based estimation results of vehicle contour target. The general parameters of ET are $d_o=30m$, $\phi_o=0^{\circ}$, $\varphi=0^{\circ}$. (a) The normalized MF output at various directions. The ET direction estimation results of (b) CRB-min Design; (c) Average-Null Design; (d) Average Design. The colors in (b)-(d) corresponds to the normalized MF value projected onto the contour element with all considered directions. The estimated direction is a straight line extended from the BS to the contour element with maximal MF value.}
\label{Fig4}
\vspace{-0.2cm}
\end{figure}

\begin{figure}[!t]
\centering
\includegraphics[width=3.5in]{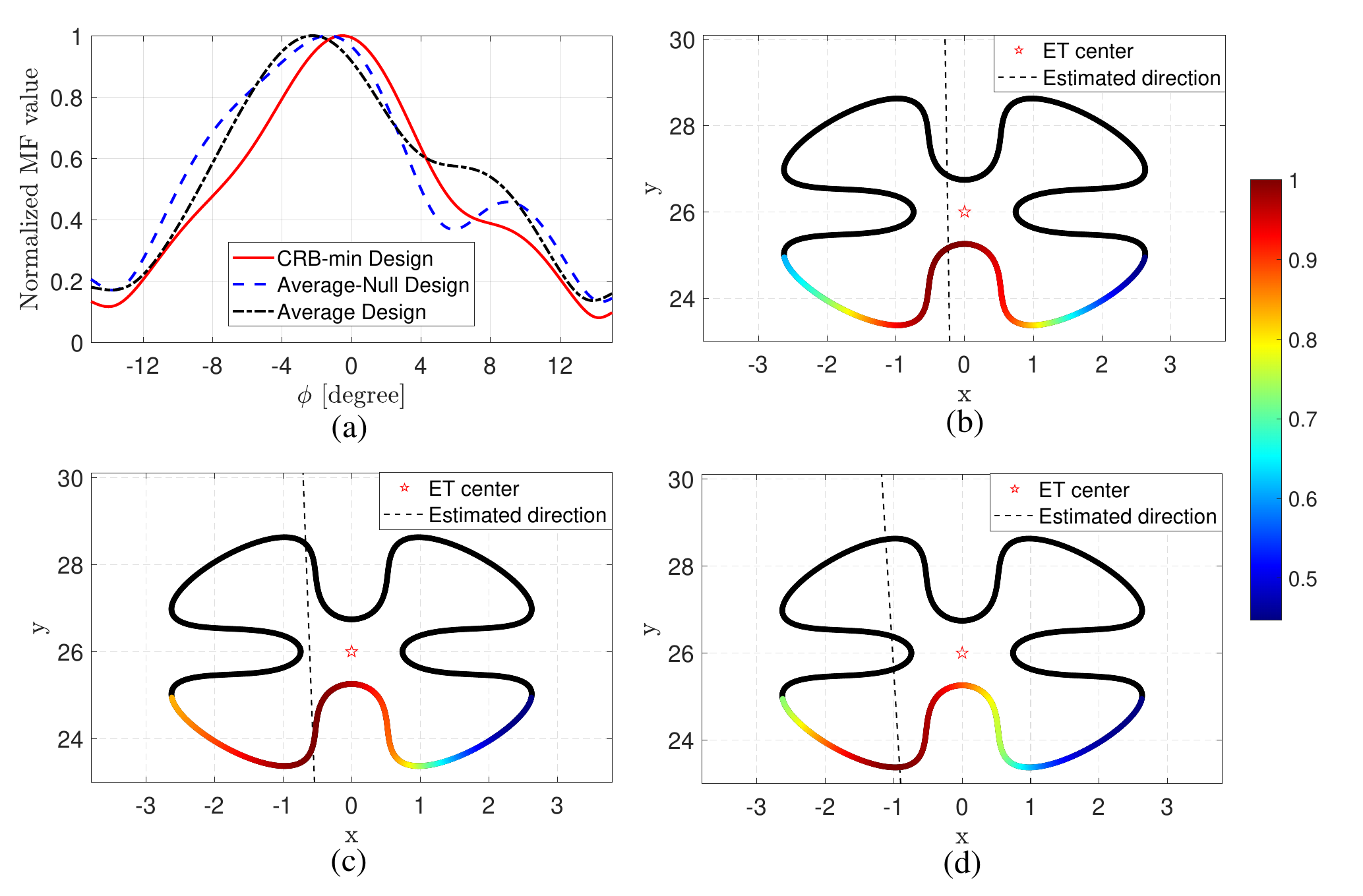}
\caption{The MF-based estimation results of UAV contour target. The general parameters of ET are $d_o=26m$, $\phi_o=0^{\circ}$, $\varphi=5^{\circ}$. (a) The normalized MF output at various directions. The ET direction estimation results of (b) CRB-min Design; (c) Average-Null Design; (d) Average Design.}
\label{Fig5}
\vspace{-0.5cm}
\end{figure}
 
 We start with the evaluation of estimating performance for ET direction, and examine the effectiveness of the corresponding CRB. A classic MF estimator is applied to extract the target direction from echo signals. Define $\mathbf{Y}=\left[ {{\mathbf{y}}_{s}}\left( 1 \right),...,{{\mathbf{y}}_{s}}\left(T \right) \right]$ as the echo signal sequence plus noise received by the BS with a certain observation period $t_s$ and a total length of $T$.\footnote{\textcolor{black}{For the ET case, it is infeasible to extract directions ${\phi_k}_{k=1}^K$ of all scatterers along the ET, contained in the channel matrix $\mathbf{H}$, from the echo signal $\mathbf{Y}=\mathbf{H}\mathbf{X}+\mathbf{N}$ if $\mathrm{rank}(\mathbf{X})=N_c< K$. Here, $\mathbf{C} \in\mathbb{C}^{N_c\times T}$ is the symbol sequence for $N_c$ CUs with length $T > N_t \geq N_c$. A common practice is to introduce extra radar dedicated probing streams in $\mathbf{x}$ \cite{Liu21TSP} and extend the degrees of freedom of $\mathbf{X}$ to its maximum, e.g., $N_t$. Nevertheless, if we only aim to estimate the center direction $\phi_o$ of the ET, it is feasible even in the case of $N_c=1$.}} Accordingly, the MF estimator can be presented as
\begin{equation}
{\hat{\phi }_{o}}=\arg \underset{\phi }{\mathop{\max }}\,\left\| {{\mathbf{a}}^{H}}\left( \phi  \right)\mathbf{Y} \right\|.
\end{equation}

The above direction estimation can be obtained via exhaustive search over a fine grid of $\phi\in\left[ -{{90}^\circ},{{90}^\circ} \right]$ with an interval of ${{0.1}^\circ}$. Fig. \ref{Fig4} and Fig. \ref{Fig5} present the target direction estimation performance of different beamforming schemes for vehicle contour and UAV contour, respectively. It can observed that our proposed CRB-min Design achieves better estimation performance over benchmark designs. The highest peaks of benchmarks obviously move leftwards or rightwards from the center, whereas the peak of our proposed CRB-min Design points straightly towards the center with the minimum deviation. Note that the estimated direction typically has significant deviation in the UAV contour case, indicating that the geometry information of ET is poorly utilized by the MF estimator.

\begin{figure}[!t]
%\centering
\includegraphics[width=3.5in]{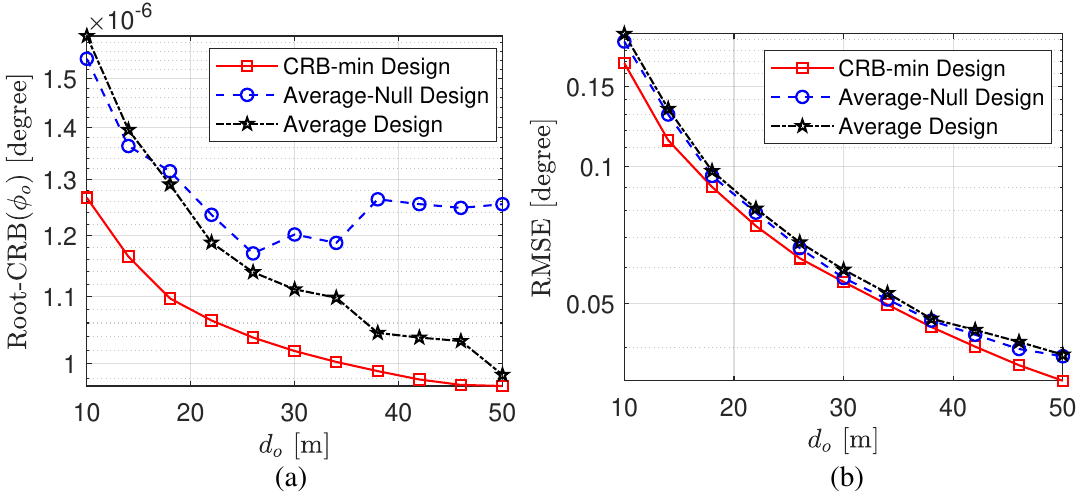}
\caption{(a) RMSE of the MF-based estimator; (b) Root-CRB of target direction parameters as a function of the distance between target and BS.}
\label{Fig6}
\vspace{-0.5cm}
\end{figure}

The CRB and MSE of direction estimation are shown in Fig. \ref{Fig6} as a function of distance, where each point is obtained based on 20000 independent Monte Carlo simulation trials. The target is assumed to move along the $+y$ axis in the global coordinate, and the radar SNR $\gamma_s$ is set constant for analysis. We can observe that CRB-min Design constantly outperforms benchmark designs with lower CRB and MSE. The CRB of Average-Null Design has random fluctuations with $d_o\geq25\text{m}$, while that of Average Design has a significant gap compared to CRB-min Design with $d_o\leq35\text{m}$. This reveals that, the unbiased direction preference of benchmark designs naturally leads to great uncertainty in direction estimation. On the other hand, our proposed CRB-min Design tends to steer the main beam towards the central direction of ET, which eliminates the potential echoes reflected from clusters elsewhere and contributes to robust direction estimation performance.

By comparing Fig. \ref{Fig6}(b) with Fig. \ref{Fig6}(a), it is seen that the direction estimation error is indeed lower bounded by the corresponding CRB, which validates the derived CRB. Nevertheless, we note that the estimation MSE is only loosely bounded by the theoretical lower bound. This suggests that MF-based estimator is not very suitable for ET and more sophisticated estimator is to be explored.
 
 Next, we examine the transmit beampatterns of three ET beamforming designs, along with the CRB-min Design for PT, as is presented in Fig. \ref{Fig7}. Considering the main-lobe characteristic for sensing an ET, our proposed CRB-min Design succeeds to steer a strong main lobe towards the ET direction, whereas the benchmark designs either steer the main-lobe towards a deviated direction from the ET center, or generate a distorted main beam with an obvious gap in the desired ET center direction. For the side-lobe characteristic, we can observe that CRB-min Design and Average Design effectively suppress the energy transmitted elsewhere, whereas Average-Null Design has severe side lobes and would eventually slow down the overall ISAC performance. \textcolor{black}{Compared with ET, the idea beampattern for PT is characterized by a sharp beam towards the target center, which is only capable of illuminating a small area.}

\begin{figure}[!t]
\centering
\includegraphics[width=3.5in]{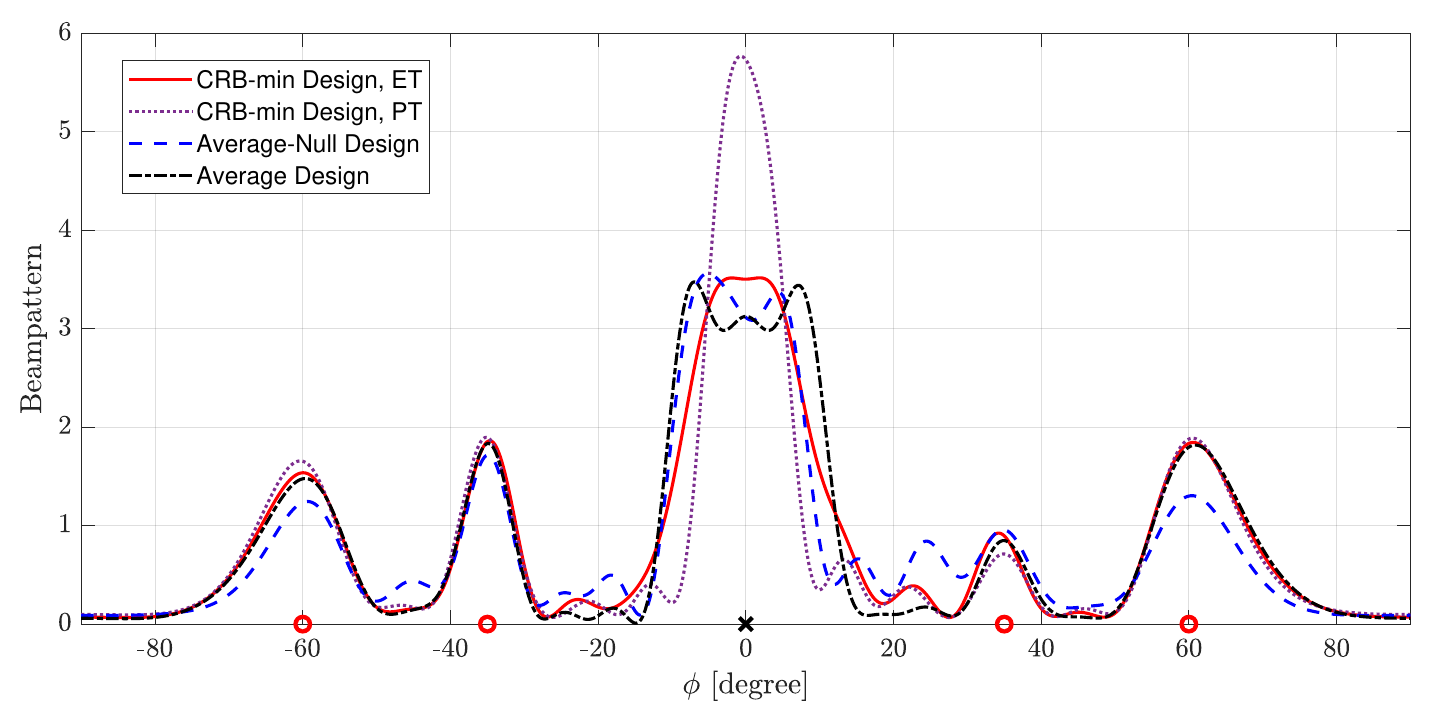}
\caption{\textcolor{black}{Beampatterns of varied ISAC beamforming designs. The red circle and black X mark refer to the directions of CUs and that of ET or PT, respectively.}}
\label{Fig7}
\vspace{-0.4cm}
\end{figure}

Moreover, Fig. \ref{Fig8} presents the CRB value versus the number of CUs. It can be noted that the sensing CRB is barely affected by communication demands in few-CU region. Nevertheless, once the number of CUs exceeds $N_c=5$, \textcolor{black}{the sensing performance can be greatly deteriorated since the power left for sensing may fail to illuminate the entire ET.} Further, with a fixed number of CUs, we observe that CRB-min Design constantly outperforms benchmark designs with lower CRBs.

\begin{figure}[!t]
\centering
\includegraphics[width=3.5in]{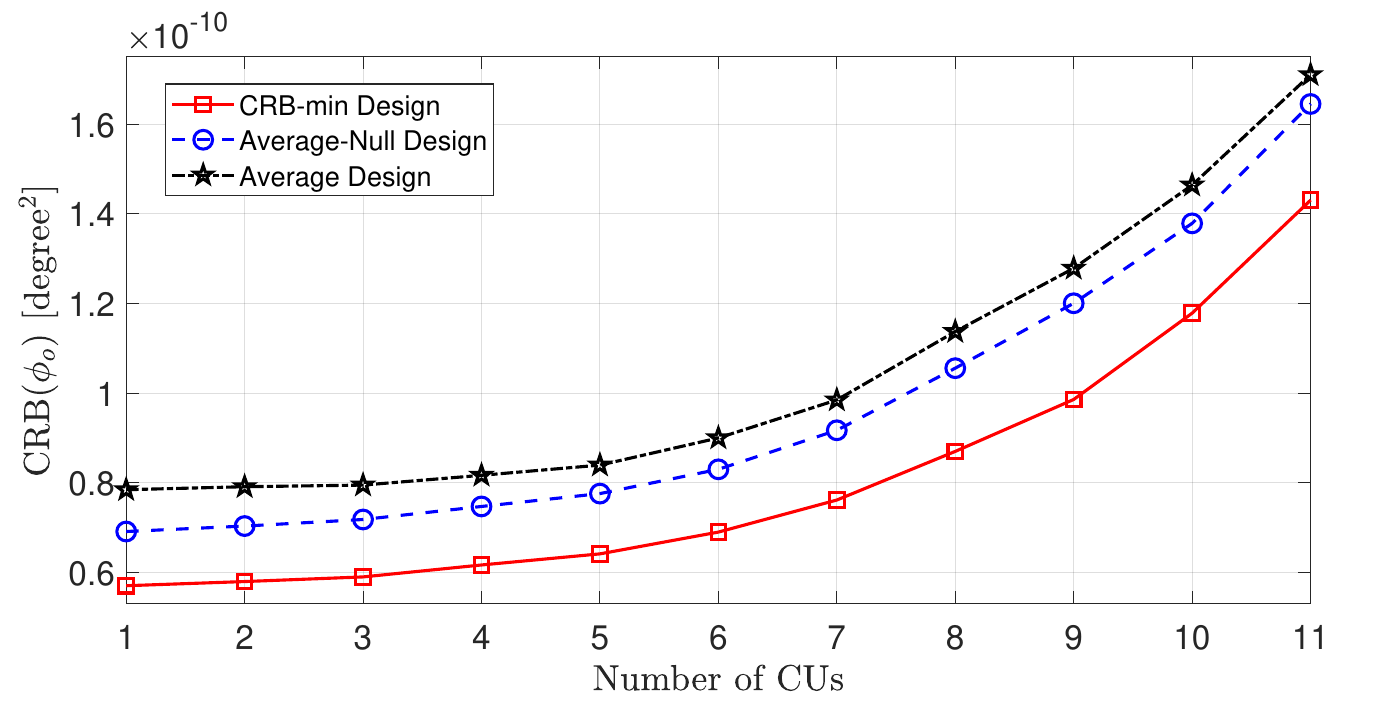}
\caption{CRBs of CRB-min and benchmark beamforming designs as a function of CU numbers.}
\label{Fig8}
\vspace{-0.2cm}
\end{figure}

As a closing remark, the fundamental difference of three beamforming designs roots in the beampattern preference. For the ET case, the beamformers should both capture the target direction and cover the whole contour. The estimation of specific directions prefers sharp beams, whereas the reception of echo signals reflected off the LoS contour requires sufficiently wide beams. The benchmark beampattern designs consider the contour coverage requirement by allocating equal energy to a given range, yet ignore the direction-capture requirement. In contrast, our proposed design achieves a balance between the sharp and wide beam contradictory, generating a desirable transmit beampattern with the minimum MSE and CRB.

\subsection{Comparison Between SDR and ZF Beamformers}
In this section, we compare the communication and sensing performance of SDR and ZF beamforming algorithms in NLoS communication channels, where the LoS link between each CU and the BS is completely blocked.

\begin{figure}[!t]
\centering
\includegraphics[width=3.5in]{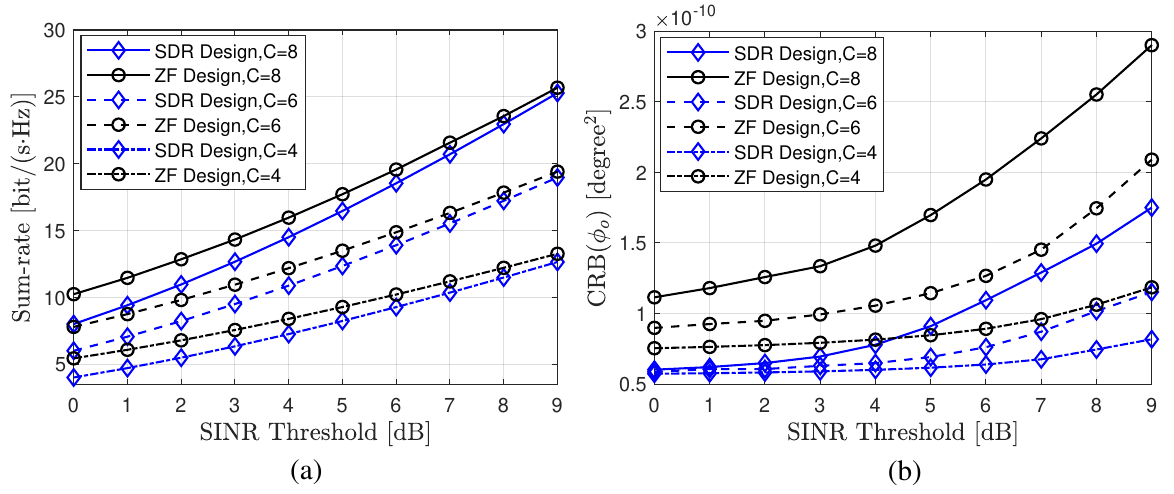}
\caption{(a) Sum rates and (b) achieved CRBs of SDR and ZF beamforming designs as a function of SINR.}
\label{Fig9}
\vspace{-0.5cm}
\end{figure}

Fig. \ref{Fig9} present the sum-rates and CRBs of the proposed two beamformers. Observed from Fig. \ref{Fig9}(a), ZF algorithm achieves higher sum-rate compared to SDR algorithm, especially in the low-SINR region. This gap gradually narrows as the SINR demand increases and approaches zero when $\Gamma =10$ dB. The above result follow the nature of ZF algorithm where ZF aims to completely eliminate the MU interference regardless of the given SINR threshold. In order to meet dual constraints of SINR threshold and interference elimination, ZF algorithm typically needs to allocate more power to communication functionality owing to the complicated multi-path channel conditions, resulting in higher sum-rates.

From Fig. \ref{Fig9}(b), we can learn that CRB increases for larger SINR threshold, indicating that better communication quality intuitively comes at the price of sensing performance loss in ISAC system. Moreover, the CRB of SDR scheme is lower than that of ZF scheme and such gap is enlarged with the increase of SINR threshold and the number of CUs. This is
consistent with our analysis in Fig. \ref{Fig9}(a), where limited power is allocated to radar sensing in NLoS channels owing to the strict communication demands.

To conclude, the low-complexity ZF beamforming scheme inevitably leads to sensing performance degradation compared to the SDR scheme. Yet, since the CRB loss is acceptable in most scenarios, we can utilize SDR beamforming algorithm under strict target estimation requirements, and flexibly switch to ZF beamforming scheme with limited computation resources.
\vspace{-0.5em}

\section{Conclusion}
This paper considers the transmit beamforming design in an MU-MIMO ISAC system, where one BS communicates with multiple CUs and senses one ET at the same time. Building on the basis of TFS contour modeling, we first derive closed-form CRBs on the central range, direction, and orientation of the ET. Then, the transmit beamformers are optimized to minimize the CRB on the target direction estimation, under the constraint of SINR threshold for communication and the beam coverage requirement for sensing. To solve the non-convex optimization problem, we first employ the SDR technique and propose a rank-one solution extraction scheme for non-tight relaxation scenarios. To reduce the calculation complexity, we further propose a ZF beamforming algorithm, which utilizes the degrees of freedom of communication channel null space to perform the sensing task.

Through numerical simulation, we explore the diverse CRB characteristics of vehicle and UAV shaped ETs, and verify the effectiveness of the derived CRB for ET. Among all circumstances considered, our proposed CRB-min design constantly achieves better estimation performance over benchmark beamforming designs, featured by lower CRB and MSE. Last but not least, compared with the our SDR beamformers, the ZF beamformers trade minor CRB loss for great computation efficiency, serving as an alternative algorithm for resource limited scenarios.\vspace{-0.4cm}
\\
{\appendices
\section*{Appendix I \\ The Proof of Theorem 1}

\textcolor{black}{Based on the derivation framework in \cite[Appendix I]{Garcia22TSP}, the overall expression for the CRB is presented in Appendix I-A, followed by the detailed derivations of the intermediate variables in Appendices I-B through I-E.}

\subsection{General Derivation of $\mathbf{J}\left( {{\boldsymbol{\kappa }}_{1}} \right)$}
Define ${{\mathbf{q}}_{k}}=\mathbf{b}\left( {{\phi }_{k}} \right){{\mathbf{a}}^{H}}\left( {{\phi }_{k}} \right)\mathbf{x}\left( t-{2{{d}_{k}}} / {c_0} \right)$, the echo signal in $(\ref{eq_ll8})$ can be expressed as $\mathbf{s}\left( t \right)=g\sum\nolimits_{k=1}^{K}{\sqrt{{l}_{k}}}{{\alpha }_{k}}{\mathbf{q}_{k}}$. Under the definition of ${{\mathbf{F}}_{{{\boldsymbol{\kappa }}_{1}}}}=-\mathbb{E}\left[ \Delta _{{{\boldsymbol{\kappa }}_{1}}}^{{{\boldsymbol{\kappa }}_{1}}}\log p\left( {{\mathbf{y}}_{s}}|\boldsymbol{\xi } \right) \right]$, we have
\begin{align}
   {{\mathbf{F}}_{{{\boldsymbol{\kappa }}_{1}}}}&=\mathbb{E}\left[ \frac{2{{g}^{2}}}{\sigma _{s}^{2}}\mathcal{R}\int_{{{t}_{s}}}{\sum\limits_{{{k}_{1}}=1}^{K}\sum\limits_{{{k}_{2}}=1}^{K}{\sqrt{{{l}_{{{k}_{1}}}}{{l}_{{{k}_{2}}}}}\alpha _{{{k}_{1}}}^{*}{{\alpha }_{{{k}_{2}}}}\frac{\partial \mathbf{q}_{{{k}_{1}}}^{H}}{\partial {{\boldsymbol{\kappa }}_{1}}}\frac{\partial {{\mathbf{q}}_{{{k}_{2}}}}}{\partial \boldsymbol{\kappa }_{1}^{T}}\mathrm{d}t}} \right] \nonumber \\ 
 & =\frac{2{{g}^{2}}}{\sigma _{s}^{2}}\mathcal{R}\int_{{{t}_{s}}}{\sum\limits_{{{k}_{1}}=1}^{K}\sum\limits_{{{k}_{2}}=1}^{K}{\mathbb{E}\left[ \alpha _{{{k}_{1}}}^{*}{{\alpha }_{{{k}_{2}}}} \right]\sqrt{{{l}_{{{k}_{1}}}}{{l}_{{{k}_{2}}}}}\frac{\partial \mathbf{q}_{{{k}_{1}}}^{H}}{\partial {{\boldsymbol{\kappa }}_{1}}}\frac{\partial {{\mathbf{q}}_{{{k}_{2}}}}}{\partial \boldsymbol{\kappa }_{1}^{T}}}}\mathrm{d}t \nonumber\\ 
 & =\frac{2{{g}^{2}}}{\sigma _{s}^{2}}\sum\nolimits_{k=1}^{K}{{{l}_{k}}}\mathcal{R}\int_{{{t}_{s}}}{\frac{\partial \mathbf{q}_{k}^{H}}{\partial {{\boldsymbol{\kappa }}_{1}}}\frac{\partial {{\mathbf{q}}_{k}}}{\partial \boldsymbol{\kappa }_{1}^{T}}\mathrm{d}t},
\end{align}
where the third equality holds since ${{\alpha }_{k}} \sim \mathcal{CN}\left( 0,1 \right)$, $\mathbb{E}\left[ \alpha _{{{k}_{1}}}^{*}{{\alpha }_{{{k}_{2}}}} \right]=1$ for ${{k}_{1}}={{k}_{2}}$ and $\mathbb{E}\left[ \alpha _{{{k}_{1}}}^{*}{{\alpha }_{{{k}_{2}}}} \right]=0$ for ${{k}_{1}}\ne {{k}_{2}}$. We notice that the calculation of $\partial \mathbf{q}_{k}^{H}/\partial {{\boldsymbol{\kappa }}_{1}}$ is complicated since ${{\mathbf{q}}_{k}}$ is directly linked with intermediate variables ${{\boldsymbol{\Theta }}_{k}}=\left[ {{d}_{k}},{{\phi }_{k}} \right]^T$ which depend on the contour parameters $\left\{ {{a}_{q}},{{b}_{q}} \right\}_{q=1}^{Q}$. Thus, following the chain rule $ \partial\mathbf{q}_{k}^{H}/ \partial{{\boldsymbol{\kappa }}_{1}}=\frac{\partial \mathbf{q}_{k}^{H}}{\partial {{\boldsymbol{\Theta }}_{k}}}\frac{\partial \boldsymbol{\Theta }_{k}^{T}}{\partial {{\boldsymbol{\kappa }}_{1}}}$, we can rewrite the FIM as follows \vspace{-0.2cm}
\begin{equation}
\label{eq51}
    {{\mathbf{F}}_{{{\boldsymbol{\kappa }}_{1}}}}=\frac{2{{g}^{2}}}{\sigma _{s}^{2}}\sum\limits_{k=1}^{K}{{{l}_{k}}}\frac{\partial \boldsymbol{\Theta }_{k}^{T}}{\partial {{\boldsymbol{\kappa }}_{1}}}\mathcal{R}\left(\int_{{{t}_{s}}}{\frac{\partial \mathbf{q}_{k}^{H}}{\partial {{\boldsymbol{\Theta }}_{k}}}\frac{\partial {{\mathbf{q}}_{k}}}{\partial \boldsymbol{\Theta }_{k}^{T}}\mathrm{d}t} \right)\frac{\partial {{\boldsymbol{\Theta }}_{k}}}{\partial \boldsymbol{\kappa }_{1}^{T}}.
\end{equation}
\vspace{-0.2cm}

With the formulas in Appendix I-B to Appendix I-E, $(\ref{eq51})$ can be further written as
\begin{align}
{{\mathbf{F}}_{{{\boldsymbol{\kappa }}_{1}}}}=&\frac{2{{g}^{2}}N_{r}}{\sigma _{s}^{2}}\sum\limits_{k=1}^{K}{{{l}_{k}}\mathbf{a}_{k}^{H}{{\mathbf{R}}_{x}}{{\mathbf{a}}_{k}}} \nonumber\\
&\left[{{\mathbf{\mu }}_{k}},{{\mathbf{\eta }}_{k}} \right]\left[ \begin{matrix}
{{Z}_{2}} & 0  \\
0 & t_s({{Z}_{1,k}}+\frac{\mathbf{\dot{a}}_{k}^{H}{{\mathbf{R}}_{x}}{{{\mathbf{\dot{a}}}}_{k}}}{\mathbf{a}_{k}^{H}{{\mathbf{R}}_{x}}{{\mathbf{a}}_{k}}})  \\
\end{matrix} \right]\left[ \begin{matrix}
\mathbf{\mu }_{k}^{T}  \\
\mathbf{\eta }_{k}^{T}  \\
\end{matrix} \right] \nonumber\\ 
 =&\frac{2{{g}^{2}}N_{r}}{\sigma _{s}^{2}}\sum\limits_{k=1}^{K}{{{l}_{k}}\mathbf{a}_{k}^{H}{{\mathbf{R}}_{x}}{{\mathbf{a}}_{k}}} \nonumber \\
&\left[ {{Z}_{2}}{{\mathbf{\mu }}_{k}}\mathbf{\mu }_{k}^{T}+\left( {{Z}_{1,k}}+\frac{\mathbf{\dot{a}}_{k}^{H}{{\mathbf{R}}_{x}}{{{\mathbf{\dot{a}}}}_{k}}}{\mathbf{a}_{k}^{H}{{\mathbf{R}}_{x}}{{\mathbf{a}}_{k}}} \right){{t_s\mathbf{\eta }}_{k}}\mathbf{\eta }_{k}^{T} \right].
\label{eq_ll29}
\end{align}

Similar to the steps in $(\ref{eq_ll29})$, we obtain ${{f}_{g}}$ and ${{\mathbf{f}}_{{{\boldsymbol{\kappa }}_{1}},g}}$ as
\begin{align}
&{{f}_{g}}=\frac{2{{N}_{r}}t_s}{\sigma _{s}^{2}}\sum\limits_{k=1}^{K}{{{l}_{k}}\mathbf{a}_{k}^{H}{{\mathbf{R}}_{x}}{{\mathbf{a}}_{k}}},\\
&{{\mathbf{f}}_{{{\boldsymbol{\kappa }}_{1}},g}} =\frac{2g}{\sigma _{s}^{2}}\sum\limits_{k=1}^{K}{{{l}_{k}}\frac{\partial \boldsymbol{\Theta }_{k}^{T}}{\partial {{\boldsymbol{\kappa }}_{1}}}}\mathcal{R} \int_{{{t}_{s}}}{\frac{\partial \mathbf{q}_{k}^{H}}{\partial {{\boldsymbol{\Theta }}_{k}}}{{\mathbf{q}}_{k}}\mathrm{d}t}  \nonumber \\
&\hspace{0.7cm}=\left[0,\hspace{0.2cm}\frac{2g{{N}_{r}}t_s}{\sigma _{s}^{2}}\sum\limits_{k=1}^{K}{{{l}_{k}}\mathcal{R}\left( \mathbf{\dot{a}}_{k}^{H}{{\mathbf{R}}_{x}}{{\mathbf{a}}_{k}} \right)},\hspace{0.2cm}0\right]^T.
\label{eq_30}
\end{align}

Finally, combining $(\ref{eq_ll29})-(\ref{eq_30})$ with $(\ref{eq_31})$, we can get the closed-form CRBs in $(\ref{CRB d_o})-(\ref{CRB varphi})$.

\subsection{Derivation of ${\partial \boldsymbol{\Theta }_{k}^{T}}/{\partial {{\boldsymbol{\kappa }}_{1}}}$}
We decompose $\partial \boldsymbol{\Theta }_{k}^{T}/\partial {{\boldsymbol{\kappa }}_{1}}=\left[ \partial {{d}_{k}}/\partial {{\boldsymbol{\kappa }}_{1}},\partial {{\phi }_{k}}/\partial {{\boldsymbol{\kappa }}_{1}} \right]$ as
\vspace{-0.3cm}
\begin{align}
{{\mathbf{\mu }}_{k}}&=\frac{\partial {{{d}}_{k}}}{\partial {{\boldsymbol{\kappa }}_{1}}}=\left[ 
   \frac{\partial {{{d}}_{k}}}{\partial {{d}_{o}}},\frac{\partial {{{d}}_{k}}}{\partial {{\phi }_{o}}},\frac{\partial {{{d}}_{k}}}{\partial \varphi }  \right]^T \nonumber \\
 &=\frac{1}{{{d}_{k}}}\Big[{{d}_{o}}+d_{o}^{-1}\boldsymbol{\rho}_{k}^{T}{{\mathbf{V}}^{T}}{{\mathbf{p}}_o},\boldsymbol{\rho}_{k}^{T}{{\mathbf{V}}^{T}}{{\mathbf{p}}_{\bot}},\boldsymbol{\rho}_{k}^{T}{{\mathbf{V}}^{T}}{{\mathbf{p}}_{\bot}}\Big]^T,\nonumber\\
&\approx \Big[ 1,\boldsymbol{\rho}_{k}^{T}{{\mathbf{V}}^{T}}{{\mathbf{p}}_{\bot }}/{{d}_{o}},\boldsymbol{\rho}_{k}^{T}{{\mathbf{V}}^{T}}{{\mathbf{p}}_{\bot }}/{{d}_{o}} \Big]^T,\\
{{\mathbf{\eta}}_{k}}&=\frac{\partial{{\phi}_{k}}}{\partial {{\boldsymbol{\kappa }}_{1}}}=\left[ 
   \frac{\partial {{\phi }_{k}}}{\partial {{d}_{o}}},\frac{\partial {{\phi }_{k}}}{\partial {{\phi }_{o}}},\frac{\partial{\phi}_{k}}{\partial \varphi} \right]^T \nonumber \\
   &=\frac{1}{d_{k}^{2}}\Big[d_{o}^{-1}\boldsymbol{\rho}_{k}^{T}{{\mathbf{V}}^{T}}{{\mathbf{p}}_{\bot }},d_{o}^{2}+\boldsymbol{\rho}_{k}^{T}{{\mathbf{V}}^{T}}{{\mathbf{p}}_{o}},\boldsymbol{\rho}_{k}^{T}{{\mathbf{V}}^{T}}{{\mathbf{p}}_{k}}\Big]^T,\nonumber\\
&\approx \Big[0,1,0\Big]^T,
\end{align}
where parameters $\left\{ {{d}_{k}},{{\phi }_{k}},{{\boldsymbol{\rho}}_{k}},{{\mathbf{p}}_{k}} \right\}$ are respectively the range, direction, local position and global position of the $k$-th contour section, ${{\mathbf{p}}_{\bot }}=\left( \begin{matrix} 0 & -1  \\ 1 & 0  \\ \end{matrix} \right){{\mathbf{p}}_{o}}$. \textcolor{black}{The approximations are valid under the assumption that $1/d_o\rightarrow 0$.}

\subsection{Derivation of $\mathcal{R}\int_{{{t}_{s}}}{\frac{\partial \mathbf{q}_{k}^{H}}{\partial {{\boldsymbol{\Theta }}_{k}}}\frac{\partial {{\mathbf{q}}_{k}}}{\partial \boldsymbol{\Theta }_{k}^{T}}\mathrm{d}t}$}
We start with the calculation of $\frac{\partial {{\mathbf{q}}_{k}}}{\partial \boldsymbol{\Theta }_{k}^{T}}=\left[ \frac{\partial {{\mathbf{q}}_{k}}}{\partial {{d}_{k}}},\frac{\partial {{\mathbf{q}}_{k}}}{\partial {{\phi }_{k}}} \right]$
\begin{align}
&\frac{\partial {{\mathbf{q}}_{k}}}{\partial {{d}_{k}}}=-\frac{2}{c}{{\mathbf{b}}_{k}}\mathbf{a}_{k}^{H}\mathbf{\dot{x}}\left( t-\frac{2{{d}_{k}}}{c_0} \right),\\
&\frac{\partial {{\mathbf{q}}_{k}}}{\partial {{\phi }_{k}}}=\left( {{{\mathbf{\dot{b}}}}_{k}}\mathbf{a}_{k}^{H}+{{\mathbf{b}}_{k}}\mathbf{\dot{a}}_{k}^{H} \right)\mathbf{x}\left( t-\frac{2{{d}_{k}}}{c_0} \right),
\end{align}
where $\mathbf{\dot{x}}\left( t \right)=\partial \mathbf{x}\left( t \right)/\partial t$, ${{\mathbf{\dot{a}}}_{k}}=\partial {{\mathbf{a}}_{k}}/\partial \phi $ and ${{\mathbf{\dot{b}}}_{k}}=\partial {{\mathbf{b}}_{k}}/\partial \phi $.

With the extra identities listed in Appendix I-E, we further derive $\mathcal{R}\int_{{{t}_{s}}}{\frac{\partial \mathbf{q}_{k}^{H}}{\partial {{\boldsymbol{\Theta }}_{k}}}\frac{\partial {{\mathbf{q}}_{k}}}{\partial \boldsymbol{\Theta }_{k}^{T}}\mathrm{d}t}$ as
\begin{align}
&\mathcal{R}\int_{{{t}_{s}}}{\frac{\partial \mathbf{q}_{k}^{H}}{\partial {{d}_{k}}}\frac{\partial {{\mathbf{q}}_{k}}}{\partial {{d}_{k}}}\mathrm{d}t}={{N}_{r}}{{Z}_{2}}\mathbf{a}_{k}^{H}{\mathbf{R}_{x}}{{\mathbf{a}}_{k}},\\
&\mathcal{R}\int_{{{t}_{s}}}{\frac{\partial \mathbf{q}_{k}^{H}}{\partial {{\phi }_{k}}}\frac{\partial {{\mathbf{q}}_{k}}}{\partial {{\phi }_{k}}}\mathrm{d}t}={{N}_{r}t_s}\left( {{Z}_{1,k}}\mathbf{a}_{k}^{H}{\mathbf{R}_{x}}{{\mathbf{a}}_{k}}+\mathbf{\dot{a}}_{k}^{H}{\mathbf{R}_{x}}{{{\mathbf{\dot{a}}}}_{k}} \right),\\
&\mathcal{R}\int_{{{t}_{s}}}{\frac{\partial \mathbf{q}_{k}^{H}}{\partial {{d}_{k}}}\frac{\partial {{\mathbf{q}}_{k}}}{\partial {{\phi }_{k}}}\mathrm{d}t}=0.
\end{align}

The complete matrix writes as
\begin{equation}
\mathcal{R}\int_{{{t}_{s}}}{\frac{\partial \mathbf{q}_{k}^{H}}{\partial {{\boldsymbol{\Theta }}_{k}}}\frac{\partial {{\mathbf{q}}_{k}}}{\partial \boldsymbol{\Theta }_{k}^{T}}\mathrm{d}t}={{N}_{r}}\mathbf{a}_{k}^{H}{\mathbf{R}_{x}}{{\mathbf{a}}_{k}}\left[ \begin{matrix}
   {{Z}_{2}} & 0  \\
   0 & t_s({{Z}_{1,k}}+\frac{\mathbf{\dot{a}}_{k}^{H}{{\mathbf{R}}_{x}}{{{\mathbf{\dot{a}}}}_{k}}}{\mathbf{a}_{k}^{H}{\mathbf{R}_{x}}{{\mathbf{a}}_{k}}})  \\
\end{matrix} \right].
\end{equation}

\subsection{Derivation of $\mathcal{R}\int_{{{t}_{s}}}{\frac{\partial \mathbf{q}_{k}^{H}}{\partial {{\boldsymbol{\Theta }}_{k}}}{{\mathbf{q}}_{k}}\mathrm{d}t}$}
Use the identities in Appendix I-E, we make following derivation
\begin{align}
&\mathcal{R}\int_{{{t}_{s}}}{\frac{\partial \mathbf{q}_{k}^{H}}{\partial {{d}_{k}}}{{\mathbf{q}}_{k}}\mathrm{d}t}=0,\\
&\mathcal{R}\int_{{{t}_{s}}}{\frac{\partial \mathbf{q}_{k}^{H}}{\partial {{\phi }_{k}}}{{\mathbf{q}}_{k}}\mathrm{d}t}={{N}_{r}t_s}\mathcal{R}\left( \mathbf{\dot{a}}_{k}^{H}{{\mathbf{R}}_{x}}{{\mathbf{a}}_{k}} \right) \nonumber \\
&\hspace{2.3cm}=\frac{{{N}_{r}t_s}}{2}\left( \mathbf{\dot{a}}_{k}^{H}{{\mathbf{R}}_{x}}{{\mathbf{a}}_{k}}+\mathbf{a}_{k}^{H}{{\mathbf{R}}_{x}}{{{\mathbf{\dot{a}}}}_{k}} \right).
\end{align}

The complete matrix writes as
\begin{equation}
\mathcal{R}\int_{{{t}_{s}}}{\frac{\partial \mathbf{q}_{k}^{H}}{\partial {{\boldsymbol{\Theta }}_{k}}}{{\mathbf{q}}_{k}}\mathrm{d}t}=\left[
   0,\hspace{0.1cm}\frac{{{N}_{r}}t_s}{2}\left( \mathbf{\dot{a}}_{k}^{H}{{\mathbf{R}}_{x}}{{\mathbf{a}}_{k}}+\mathbf{a}_{k}^{H}{{\mathbf{R}}_{x}}{{{\mathbf{\dot{a}}}}_{k}} \right)\right]^T.
\end{equation}

\subsection{Related Identities}
To facilitate the calculations, we take the center of the ULA as the reference point with zero phase. Thus, the transmit and receive antenna response can be presented as 
\begin{align}
&{{\mathbf{a}}_{k}}=\exp \left( j\pi \left( {{N}_{t}}-1 \right)/2\sin {{\phi }_{k}} \right) \nonumber\\
&\hspace*{.8cm}{{\Big[ 1,\exp\left( -j\pi \sin {{\phi }_{k}} \right),...,\exp\left( -j\left( {{N}_{t}}-1 \right)\pi \sin {{\phi }_{k}} \right) \Big]}^{T}},
\end{align}
\setlength\abovedisplayskip{1pt}
\setlength\belowdisplayskip{1pt}
\begin{align}
&{{\mathbf{b}}_{k}}=\exp \left( j\pi \left( {{N}_{r}}-1 \right)/2\sin {{\phi }_{k}} \right) \nonumber\\
&\hspace*{.8cm}{{\Big[ 1,\exp\left( -j\pi \sin {{\phi }_{k}} \right),...,\exp\left( -j\left( {{N}_{r}}-1 \right)\pi \sin {{\phi }_{k}} \right) \Big]}^{T}}.
 \end{align}

We can further calculate that ${{\left\| {{\mathbf{b}}_{k}} \right\|}^{2}}={{N}_{r}}$, and get
\begin{equation}
{{\mathbf{\dot{a}}}_{k}}=j\pi \cos \left( {{\phi }_{k}} \right)\text{diag}\left( \frac{{{N}_{t}}-1}{2},...,-\frac{{{N}_{t}}-1}{2} \right){{\mathbf{a}}_{k}},
\end{equation}
\begin{equation}
{{\mathbf{\dot{b}}}_{k}}=j\pi \cos \left( {{\phi }_{k}} \right)\text{diag}\left( \frac{{{N}_{r}}-1}{2},...,-\frac{{{N}_{r}}-1}{2} \right){{\mathbf{b}}_{k}},
\end{equation}
\begin{equation}
{{\left\| {{{\mathbf{\dot{b}}}}_{k}} \right\|}^{2}}=\frac{{{\cos }^{2}}\left( {{\phi }_{k}} \right){{\pi }^{2}}\left( N_{r}^{3}-{{N}_{r}} \right)}{12}={{N}_{r}}{{Z}_{1,k}}.
\end{equation}

Next, with the basic assumption in $(\ref{eq2})$, it is clear that $\int_{t_s} \mathbf{c}(t) \mathbf{c}^H (t) \mathrm{d}t = t_s\mathbf{I}_{N_c}$. Since we are considering a sufficient long observation period $t_s$, the following Fourier transform derivation holds when the integration is calculated within $t_s$
\begin{align}
  & \int_{t_s}{c_{{{n}_{1}}}^{*}\left( t-\tau  \right)\dot{c}_{{n}_{2}}\left( t-\tau  \right)\mathrm{d}t} \nonumber\\ 
 & \approx\int_{-\infty }^{+\infty }{{{\left[ {{C}_{{{n}_{1}}}}\left( f \right){{e}^{-j2\pi \tau f}} \right]}^{*}}j2\pi f{{C}_{{{n}_{2}}}}\left( f \right){{e}^{-j2\pi \tau f}}\mathrm{d}f} \nonumber\\ 
 & =j2\pi \int_{-\infty }^{+\infty }{fC_{{{n}_{1}}}^{*}\left( f \right){{C}_{{{n}_{2}}}}\left( f \right)\mathrm{d}f}=0, n_{1,2}=1,...,N_c.
 \label{eq_33}
\end{align}

For ${{n}_{1}}\ne {{n}_{2}}$, $(\ref{eq_33})$ holds since $\int_{-\infty }^{+\infty }{fC_{{{n}_{1}}}^{*}\left( f \right){{C}_{{{n}_{2}}}}\left( f \right)\mathrm{d}f}=0$ given the irrelevance between ${{C}_{{{n}_{1}}}}\left( f \right)$ and ${{C}_{{{n}_{2}}}}\left( f \right)$. For ${{n}_{1}}={{n}_{2}}=n$, the integral term $\int_{-\infty }^{+\infty }{f{{\left| {{C}_{n}}\left( f \right) \right|}^{2}}\mathrm{d}f}$ can be regarded as the mass center of the signal spectrum ${{\left| {{C}_{n}}\left( f \right) \right|}^{2}}$. Since such center can be shifted arbitrarily in the frequency domain, we can locate it at a specific point with zero value such that $(\ref{eq_33})$ still holds. Further, we have

\begin{align}
\label{eq_34}
&\int_{t_s}{\dot{c}_{{n}_{1}}^{*}\left( t-\tau  \right){\dot{c}_{{n}_{2}}}\left( t-\tau  \right)\mathrm{d}t} \nonumber\\\nonumber
&\approx{{\left( 2\pi  \right)}^{2}}\int_{-\infty }^{+\infty }{{f^{2}}S_{{n}_{1}}^{*}\left( f \right)S_{{n}_{2}}\left( f \right)\mathrm{d}f}\\
&=\left\{
\begin{aligned}
&{\left( 2\pi B \right)}^{2}, && {n}_{1}={n}_{2}  \\
&0, && {n}_{1}\ne {n}_{2}
\end{aligned}
\right.,
\quad n_{1,2}=1,...,N_c.
\end{align}

Similar with $(\ref{eq_33})$, for ${{n}_{1}}\ne {{n}_{2}}$, ${{C}_{{{n}_{1}}}}\left( f \right)$ is irrelevant with ${{C}_{{{n}_{2}}}}\left( f \right)$ so that $(\ref{eq_34})$ holds. For ${{n}_{1}}= {{n}_{2}}=n$, $(\ref{eq_34})$ follows the definition of $B=\left(\int_{-\infty }^{+\infty }{{{f}^{2}}{{\left| {{C}_{n}}\left( f \right) \right|}^{2}}\mathrm{d}f}\right)^{1/2}$.


\end{document}